\documentclass[12pt]{article}
\usepackage[colorlinks, linkcolor=red, citecolor=blue]{hyperref}
\usepackage{amsmath}
\usepackage{amssymb}
\usepackage{float}
\usepackage{mathtools}     
\usepackage{times}
\usepackage{bm}
\usepackage{natbib}
\usepackage{colonequals}
\usepackage{graphicx}
\usepackage{threeparttable}
\usepackage{dcolumn}
\usepackage{multirow}
\usepackage{caption}
\usepackage{color}
\usepackage{authblk}
\usepackage{enumerate}
\usepackage{cmll}
\usepackage{booktabs}
\usepackage{geometry}
\usepackage{setspace}
\usepackage{amsthm}
\usepackage{authblk}

\usepackage{algorithm}
\usepackage{algpseudocode}

\usepackage{mathtools}     
\usepackage{color}

\newtheorem{theorem}{\bf{Theorem}}

\newtheorem{lemma}{\bf{Lemma}}

\newtheorem{corollary}{\bf{Corollary}}

\makeatletter
\newenvironment{breakablealgorithm}
{
		\begin{center}
			\refstepcounter{algorithm}
			\hrule height.8pt depth0pt \kern2pt
			\renewcommand{\caption}[2][\relax]{
				{\raggedright\textbf{\ALG@name~\thealgorithm} ##2\par}%
				\ifx\relax##1\relax 
				\addcontentsline{loa}{algorithm}{\protect\numberline{\thealgorithm}##2}%
				\else 
				\addcontentsline{loa}{algorithm}{\protect\numberline{\thealgorithm}##1}%
				\fi
				\kern2pt\hrule\kern2pt
			}
		}{
		\kern2pt\hrule\relax
	\end{center}
}
\makeatother

\usepackage{lipsum}
\usepackage{arydshln}
\usepackage{subfig}
\graphicspath{{figures/}} 	

\makeatletter  

\newcommand{\var}{{\rm var}}
\newcommand{\bX}{\text{\boldmath{$X$}}}
\newcommand{\Cov}{\mbox{Cov}}

\newcommand{\btheta}{\text{\boldmath{$\theta$}}}

\newcommand{\bbM}{\mathbf{M}}
\newcommand{\bbH}{\mathbf{H}}
\newcommand{\bbU}{\mathbf{U}}
\newcommand{\bbV}{\mathbf{V}}

\newcommand{\bz}{\boldsymbol{z}}

\newcommand{\bbeta}{\boldsymbol{\beta}}

\newcommand{\bg}{\boldsymbol{g}}

\newcommand{\bJ}{\boldsymbol{J}}

\newcommand{\bU}{\boldsymbol{U}}
\newcommand{\bZ}{\boldsymbol{Z}}
\newcommand{\bzero}{\boldsymbol{0}}

\newcommand{\bK}{{\boldsymbol K}}
\newcommand{\bV}{\boldsymbol{V}}

\newcommand{\bbA}{\mathbf{A}}

\newcommand{\bv}{\boldsymbol{v}}

\makeatother

\def\T{{ \mathrm{\scriptscriptstyle T} }}

\def\argmin{\mathop{\rm argmin}}

\date{}
\begin{document}
	\title{Subsampled One-Step Estimation for Fast Statistical Inference}
	\author[1]{Miaomiao Su}
	\author[2]{Ruoyu Wang\thanks{Corresponding author: ruoyuwang@hsph.harvard.edu}}
	\affil[1]{School of Mathematical Sciences,  Beijing University of Posts and Telecommunications, Beijing, China, and Key Laboratory of Mathematics andInformation Networks (Beijing University of Posts and Telecommunications), Ministry of Education, China.}
	\affil[2]{Department of Biostatistics, Harvard University, Boston, Massachusetts, USA}
	
	\maketitle
	
	\begin{abstract}
		
		Subsampling is an effective approach to alleviate the computational burden associated with large-scale datasets. Nevertheless, existing subsampling estimators incur a substantial loss in estimation efficiency compared to estimators based on the full dataset. Specifically, the convergence rate of existing subsampling estimators is typically $n^{-1/2}$ rather than $N^{-1/2}$, where $n$ and $N$ denote the subsample and full data sizes, respectively. This paper proposes a subsampled one-step (SOS) method to mitigate the estimation efficiency loss through a one-step update based on the asymptotic expansions of the subsampling and full-data estimators. The resulting SOS estimator is computationally efficient and achieves a fast convergence rate of $\max(n^{-1}, N^{-1/2})$ rather than $n^{-1/2}$. We establish the asymptotic distribution of the SOS estimator,  which can be non-normal in general, and construct confidence intervals on top of the asymptotic distribution. Furthermore, we prove that the SOS estimator is asymptotically normal and equivalent to the full data-based estimator when $n / \sqrt{N} \to \infty$. Simulation studies and real data analyses were conducted to demonstrate the finite sample performance of the SOS estimator. Numerical results suggest that the SOS estimator is almost as computationally efficient as the uniform subsampling estimator while achieving estimation efficiency similar to the full data-based estimator.
		\vspace*{0.3em}
		
		\noindent {\it Keywords: Large-scale data; One-step estimation; Statistical efficiency; Subsampling.}
		
	\end{abstract}
	
	\doublespacing
	\section{Introduction}\label{sec1}
	
	In recent years, big data has gained significant attention across various fields, offering statisticians unprecedented opportunities to analyze vast amounts of information. However, traditional statistical methods can be computationally cumbersome or even infeasible when applied to large-scale datasets. Existing methods for processing big data can be divided into three categories: distributed methods \citep{Mcdonald2009dist,Zhang2013dist,Lee2017dist,jordan2019communication}, online update methods \citep{lin2011aggregated,schifano2016online,luo2020renewable}, and subsampling methods \citep{Drines2006Subsampling,Drines2011Subsampling,Fithian2014Subsampling,Ma2015aSubsampling,Ma2015bSubsampling,Han2020Subsampling,Wang2018Subsampling,Wang2019Subsampling,Yao2019Subsampling,Wang2020Subsampling,Yu2020Subsampling,Ai2019Subsampling,Ai2021Subsampling}.	
	Distributed computing methods address computational challenges by assigning tasks to multiple machines, computing on each machine, and aggregating results of different local machines to get a final estimate. These methods can be divided into one-shot and multi-round approaches. One-shot methods can achieve the same estimation efficiency as the full-data estimator when the number of machines is not large and may suffer from efficiency loss otherwise \citep{chen2022first}. Multi-round methods, on the other hand, relax the restriction on the machine number via iterative refinement but may incur higher communication cost \citep{fan2023communication}. Federated learning, a specialized form of distributed learning, further extends this idea by allowing model training across decentralized datasets without sharing raw data, thus addressing privacy concerns in big data analysis \citep{li2020review,zhang2021survey,li2024efficient}. Although effective, distributed computing methods generally require parallel computing infrastructure, imposing significant hardware and network requirements. Online update methods provide another alternative by continuously updating model estimates with newly available data, without storing historical information. It can effectively alleviate computational and storage challenges. However, online updating is primarily designed for streaming data scenarios where observations arrive sequentially, such as in real-time monitoring applications \citep{luo2023statistical}, and may not be suitable for large-scale batch processing \citep{mahendran2023model}.
	Compared to distributed and online updating algorithms, subsampling techniques have become prevalent in numerous real-world applications due to their computational efficiency and their ability to be implemented with minimal computational resources, e.g., a single laptop. This article mainly focuses on the subsampling strategy.

	The main issue of the subsampling strategy is that it only uses a small part of the data, making it hard to attain the estimation efficiency of the full data-based estimator \citep{Drines2006Subsampling,kushilevitz_nisan_1996}. Therefore, many works are devoted to improving the estimation efficiency of subsampling estimation. The rationale behind many existing subsampling improvement strategies is that the information contained in different observations is different, and observations with more information should be sampled with higher probability \citep{Wang2019Efficient}. Based on this idea, many improvement schemes focus on designing nonuniform sampling probability (NSP) to improve the estimation efficiency of subsampling estimators.
	\cite{Drines2006Subsampling} proposed a leverage subsampling method for linear models, which prefers to select leverage points. \cite{Drines2011Subsampling} and \cite{Ma2015aSubsampling} proposed some fast calculation algorithms to speed up the calculation of the leverage subsampling method. \cite{Fithian2014Subsampling} proposed a local case-control subsampling method for logistic regression, which prefers to keep data points that are easily misclassified. \cite{Han2020Subsampling} extended the idea in \cite{Fithian2014Subsampling} to the multi-class classification problem. These methods design subsampling probability based on some intuitive sample importance. Another line of work is the variance-based subsampling methods that design the NSP to minimize some criterion function of the resulting estimator's asymptotic variance \citep{Wang2018Subsampling}.
	This idea has been developed to deal with various problems including the multi-class logistic regression \citep{Yao2019Subsampling}, the quantile regression \citep{Ai2021Subsampling}, the quasi-likelihood \citep{Yu2020Subsampling}, and cox regression \citep{keret2023analyzing}.
	Recently, \cite{fan2022nearly} proposes an empirical likelihood weighting method that improves the estimation efficiency of the subsampling estimator by incorporating some easily available auxiliary information.
	The above estimators converge to the true parameter at the rate $n^{-1/2}$, where $n$ is the subsample size. Typically, $n$ is much smaller than the sample size of the full data $N$. Thus, these subsampling estimators still have substantial estimation efficiency loss compared to the full data-based estimator. 

	This paper proposes a new approach to mitigating the estimation efficiency loss in subsampling M-estimator. The proposed method is based on uniform subsampling which includes different observations with equal probability.
	We commence by investigating the asymptotic expansions of the uniformly subsampled and full data-based estimators, subsequently devising a {\it{subsampled one-step}} (SOS) method to improve the estimation efficiency of the uniformly subsampled estimator while maintaining the computational efficiency.  
	We prove that the SOS estimator has a fast convergence rate $\max(n^{-1},N^{-1/2})$ rather than $n^{-1/2}$. We establish the asymptotic distribution of the proposed estimator which can be non-normal in general. The confidence interval is provided for statistical inference based on the asymptotic distribution of SOS. Furthermore, we prove that the SOS estimator is asymptotically as efficient as the full data-based estimator when $n/\sqrt{N}\to \infty$. 
	Extensive numerical experiments were conducted based on both simulated and real datasets to demonstrate the finite sample performance of the proposed estimator. The proposed estimator exhibits promising performance in terms of both estimation efficiency and computational efficiency in our numerical results.
	
	It is worth noting that the one-step update is a long-standing idea that dates back at least to \cite{le1956asymptotic}. In recent years, this approach has continued to find success in a variety of statistical problems.  For instance, \citet{Kamatani2015Hybrid} leveraged it to stabilize the simultaneous maximum likelihood estimator for the estimation problem under a parametric diffusion process. \cite{Brouste2023Fast}, \cite{hariz2023fast}, and \cite{hariz2024fast} constructed one-step estimators for parameter estimation in Hawkes processes, weak fractionally autoregressive integrated moving-average models, and first-order fractional autoregressive models, respectively, to accelerate computation while preserving full efficiency. \cite{Kutoyants2016OnMM} proposed multiple-step estimators for parametric Markov processes, iteratively refining an initial estimator until it achieves the asymptotic efficiency of the full-data maximum likelihood estimator.
		These prior works primarily focus on maximum likelihood estimation or least squares estimation under parametric models. In contrast, our results showcase the power of the one-step update in the context of general M-estimation with large-scale data and differ from existing results in several key aspects. Computationally, our SOS method uses the Hessian matrix computed solely from the subsample, which improves computational efficiency with minimal loss in statistical accuracy. Theoretically, this paper contributes by establishing the asymptotic distribution of the one-step estimator in the regime $n/\sqrt{N}\nrightarrow\infty$, and providing a valid inference procedure based on the limiting distribution. This result is particularly valuable in large-scale data settings, where the full sample size $N$ is extremely large and a small subsample size $n$ is essential for computational tractability.
		This aspect distinguishes our work from earlier literature on one-step updates, such as \cite{Kutoyants2016OnMM}, which discussed the convergence rate $\max(n^{-1},N^{-1/2})$ for one-step estimators but did not derive the corresponding asymptotic distribution.
	
	The rest of this paper is organized as follows. In Section \ref{sec: method}, we introduce the SOS method. In Section \ref{sec: properties}, we establish the asymptotic properties of the SOS estimator. In Section \ref{sec: simulation}, we conduct simulation studies to evaluate the finite sample performance of the SOS estimator in terms of estimation efficiency and computation efficiency. Furthermore, we present a real data example in Section \ref{sec: real} to illustrate the practical application of the SOS method. The proofs are delayed to the supplementary material.

	\section{Methodology}\label{sec: method}
	
	Let $\{\bZ_{i}\}_{i=1}^{N}$ be independent and identically distributed (i.i.d.) observations of a random vector $\bZ$.
	Suppose the parameter of interest $\btheta_{0}\in \mathbb{R}^{d}$ is the minimizer of the population-level loss function $E\{L(\bZ;\btheta)\}$, that is,
	\begin{equation*}
		\btheta_{0} = \argmin_{\btheta\in\Theta}E\{L(\bZ;\btheta)\},
	\end{equation*}
	where $L$ is a loss function and $\Theta$ is the parameter space. A natural estimator for $\btheta_{0}$ can be obtained by minimizing the sample version of $E\{L(\bZ;\btheta)\}$.
	Specifically, the full data-based M-estimator is
	\begin{equation}\label{eq:full est}
		\hat{\btheta}_{\rm full} = \argmin_{\btheta\in\Theta}\frac{1}{N}\sum_{i=1}^{N} L(\bZ_{i};\btheta).
	\end{equation}
	Computation of the full data-based estimator $\hat{\btheta}_{\rm full}$ is time-consuming when dealing with extremely large datasets, particularly in cases where $\hat{\btheta}_{\rm full}$ does not have a closed-form solution and requires to be solved by an iterative algorithm.  The computing time of $\hat{\btheta}_{\rm full}$ is of order $O\{\xi(Nd^{2}+d^{3}+Nd)\}$ when employing the Newton-Raphson algorithm, where $\xi$ denotes the number of iterations. 
	When the sample size is large, storing the data and computing the full data-based estimator $\hat{\btheta}_{\rm full}$ can be challenging or even infeasible.
	
	Subsampling provides an effective approach to reducing the computational burden associated with large-scale data. The simplest subsampling is the uniform subsampling. The uniform subsampling method involves generating a Bernoulli random variable $R_{i}$ for each data point $\bZ_{i}$ with $P(R_{i}=1\mid \bZ_{i})=n/N$, where $n$ represents the expected subsample size. Let $S=\{i\mid R_{i}=1,i\in\{1,\dots,N\}\}$ denote the index set of the subsample. Based on the subsample indexed by $S$, the uniform subsampling M-estimator can be defined as:
	\begin{equation}\label{eq:uni est}
		\tilde{\btheta}_{\rm uni} = \argmin_{\btheta\in\Theta}\sum_{i\in S}L(\bZ_{i};\btheta).
	\end{equation}
	The uniform subsampling estimator $\tilde{\btheta}_{\rm uni}$ is computationally efficient, while it is not as accurate as the full data-based estimator $\hat{\btheta}_{\rm full}$ since only a small portion of the data is used. Extensive studies devote to improving the estimation efficiency of $\tilde{\btheta}_{\rm uni}$ \citep{Drines2006Subsampling,Yu2020Subsampling,Yao2019Subsampling,Wang2018Subsampling,Wang2019Subsampling,Wang2020Subsampling,Ai2019Subsampling}.
	However, the convergence rate of existing subsampling estimators is typically $n^{-1/2}$, which can be much slower than that of $\hat{\btheta}_{\rm full}$.	
	
	This paper aims to develop a computationally and statistically efficient estimator that can asymptotically achieve the same estimation efficiency as the full data-based estimator $\hat{\btheta}_{\rm full}$ while retaining similar computing time as $\tilde{\btheta}_{\rm uni}$.
	Note that under regularity conditions, the full data-based estimator $\tilde{\btheta}_{\rm full}$ and the uniform subsampling estimator $\tilde{\btheta}_{\rm uni}$ have the following asymptotic expansions 
	\begin{equation}\label{eq:full expan}
		\tilde{\btheta}_{\rm full}-\btheta_{0} = - \left[E\left\{\nabla^{2}L(\bZ;\btheta_{0})\right\}\right]^{-1}\frac{1}{N}\sum_{i=1}^{N}\nabla L(\bZ_{i};\btheta_{0}) + O_{P}(N^{-1})
	\end{equation}
	and 
	\begin{equation}\label{eq:uni expan}
		\tilde{\btheta}_{\rm uni}-\btheta_{0} = - \left[E\left\{\nabla^{2}L(\bZ;\btheta_{0})\right\}\right]^{-1}\frac{1}{n}\sum_{i\in S}\nabla L(\bZ_{i};\btheta_{0}) + O_{P}(n^{-1}),
	\end{equation}
	respectively.
	By comparing the expansions in \eqref{eq:full expan} and \eqref{eq:uni expan}, we have
	\begin{equation}\label{eq:difference}
		\tilde{\btheta}_{\rm full} - \tilde{\btheta}_{\rm uni}= - \left[E\left\{\nabla^{2}L(\bZ;\btheta_{0})\right\}\right]^{-1}\left\{\frac{1}{N}\sum_{i=1}^{N}\nabla L(\bZ_{i};\btheta_{0})-\frac{1}{n}\sum_{i\in S}\nabla L(\bZ_{i};\btheta_{0})\right\}+O_{P}(n^{-1}). 
	\end{equation}
	One can find that the second term at the right side of \eqref{eq:difference} is ignorable when $n/\sqrt{N}\to\infty$, i.e., the subsample size is not too small.
	The estimation efficiency loss comes mainly from the first term on the right side of \eqref{eq:difference}, which is of order $O_{P}(n^{-1/2})$. 
	
	To mitigate the estimation efficiency loss while maintaining computational efficiency, we construct an easy-to-compute estimate for the first term on the right side of \eqref{eq:difference}. Specifically, we estimate $\left[E\left\{\nabla^{2}L(\bZ;\btheta_{0})\right\}\right]^{-1}$ and $N^{-1}\sum_{i=1}^{N}\nabla L(\bZ_{i};\btheta_{0})-n^{-1}\sum_{i\in S}\nabla L(\bZ_{i};\btheta_{0})$ by $\left\{n^{-1}\sum_{i\in S}\nabla^2 L(\bZ_{i};\tilde{\btheta}_{\rm uni})\right\}^{-1}$ and
	\[
	\begin{aligned}
		&\frac{1}{N}\sum_{i=1}^{N}\nabla L(\bZ_{i};\tilde{\btheta}_{\rm uni})-\frac{1}{n}\sum_{i\in S}\nabla L(\bZ_{i};\tilde{\btheta}_{\rm uni})
		= \frac{1}{N}\sum_{i=1}^{N}\nabla L(\bZ_{i};\tilde{\btheta}_{\rm uni}),
	\end{aligned}
	\]
	respectively. Then, we define the SOS estimator as
	\begin{equation*}
		\tilde{\btheta}_{\rm SOS} = \tilde{\btheta}_{\rm uni} - \left\{\frac{1}{n}\sum_{i\in S}\nabla^2 L(\bZ_{i};\tilde{\btheta}_{\rm uni})\right\}^{-1}\frac{1}{N}\sum_{i=1}^{N}\nabla L(\bZ_{i};\tilde{\btheta}_{\rm uni}),
	\end{equation*} 
	Notice that the computation of $N^{-1}\sum_{i=1}^{N}\nabla L(\bZ_{i};\tilde{\btheta}_{\rm uni})$ only involves sample average over the full data, which is easy to calculate even for large-scale data. The time complexity for computing $N^{-1}\sum_{i=1}^{N}\nabla L(\bZ_{i};\tilde{\btheta}_{\rm uni})$ is $O(Nd)$ that is of the same order as that of computing the NSP for existing subsampling methods \citep{Wang2019Efficient, Wang2020Subsampling}.

	The form of the SOS estimator resembles the one-step update from $\tilde{\btheta}_{\rm uni}$ using Newton's method, which is given by $\tilde{\btheta}_{\rm uni}-\left\{N^{-1}\sum_{i=1}^{N}\nabla^2 L(\bZ_{i};\tilde{\btheta}_{\rm uni})\right\}^{-1}N^{-1}\sum_{i=1}^{N}\nabla L(\bZ_{i};\tilde{\btheta}_{\rm uni})$.
	The initial estimator $\tilde{\btheta}_{\rm uni}$ is consistent and can be easily calculated.
	In addition, the computation of the Hessian in the SOS estimator is based solely on the subsample which further alleviates the computational problem when dealing with large datasets.
	The one-step update in the SOS estimator enhances the estimation efficiency of the uniform subsampling estimator by incorporating more comprehensive information from the full dataset. This allows it to converge faster than many existing subsampling methods and maintain computational feasibility by avoiding full dataset iterative computations.
	To facilitate practitioners, we provide a step-by-step computation in Algorithm \ref{general}.
	
	\begin{breakablealgorithm}
		\caption{Subsampled One-Step Algorithm.}
		\begin{algorithmic}
			\State \textbf{Step 1: Subsampling.} Initialize $S=\varnothing$. For $i=1,\dots,N$, generate a Bernoulli variable $R_{i}\sim$ Bernoulli$(n/N)$, and if $R_{i} = 1$, update $S= S\cup\{i\}$.
			\State \textbf{Step 2: Initial Estimation.} Based on the subsample indexed by $S$, calculate the uniform subsampling M-estimator 
			\begin{equation*}
				\tilde{\btheta}_{\rm uni} = \argmin_{\btheta\in\Theta}\sum_{i\in S}L(\bZ_{i};\btheta).
			\end{equation*}
			For example, let $Y$ be an outcome variable, $\bX$ a vector of covariates, and $\bZ = (Y,\bX^{\T})^{\T}$. Then $L(\bZ;\btheta) = (Y-\bX^{\T}\btheta)^{2}$ under linear regression and 
			$$L(\bZ;\btheta) = -\left[Y\bX^{\T}\btheta-\log\left\{1+\exp(\bX^{\T}\btheta)\right\}\right]$$ under logistic regression.
			\State \textbf{Step 3: Compute the Hessian and Gradient.} Based on the initial estimator $\tilde{\btheta}_{\rm uni}$ in Step 2 and  the subsample indexed by $S$, calculate the Hessian 
			$$\widetilde{\bbH} = \frac{1}{n}\sum_{i\in S}\nabla^2 L(\bZ_{i};\tilde{\btheta}_{\rm uni}).$$
			Based on the initial estimator $\tilde{\btheta}_{\rm uni}$ in Step 2 and the full data, calculate the gradient $$\widetilde{\bJ}=\frac{1}{N}\sum_{i=1}^{N}\nabla L(\bZ_{i};\tilde{\btheta}_{\rm uni}).$$
			\State \textbf{Step 4: One-Step Update.} Based on the initial estimator $\tilde{\btheta}_{\rm uni}$ in Step 2 and the Hessian $\widetilde{\bbH}$ and gradient $\widetilde{\bJ}$ in Step 3, calculate the SOS estimator
			\begin{equation*}
				\tilde{\btheta}_{\rm SOS} = \tilde{\btheta}_{\rm uni} - \widetilde{\bbH}^{-1}\widetilde{\bJ}.
			\end{equation*} 
		\end{algorithmic}
		\label{general}
	\end{breakablealgorithm}

	The one-step update is a classical idea that dates back at least to \cite{le1956asymptotic}, and has led to fruitful developments across a wide range of statistical problems. For example, it is employed in semiparametric problems to construct efficient estimators \citep{bickel1993efficient} and applied to resolve the numerical difficulty in solving a complicated estimating equation \citep{Vaart2000AS}. 
	More recently, the one-step update strategy has also been leveraged to alleviate computational burdens. For instance, \citet{hariz2024fast} extended the one-step procedure to estimate parameters in first-order fractional autoregressive models. \citet{Brouste2023Fast} and \citet{hariz2023fast} constructed subsample-based one-step estimators for Hawkes processes and weak fractionally autoregressive integrated moving-average models, respectively, achieving computational gains while retaining statistical efficiency. \citet{Kutoyants2016OnMM} proposed multiple-step estimators for parametric Markov processes, iteratively refining a subsample-based initial estimator until it attained the asymptotic efficiency of the full-data estimator. Our SOS method and the above methods all achieve high computational and statistical efficiency using the one-step update idea. 	
	Nonetheless, the SOS method differs from existing research in terms of its scope of application, estimator construction, and theoretical properties. Existing methods largely focus on maximum likelihood estimation or least squares estimation under parametric models, whereas we consider a more general M-estimation setting that does not require parametric assumption on the data distribution and contains maximum likelihood estimation and least squares estimation as special cases. Additionally, the parametric assumptions adopted in the previous works imply that the Hessian matrix is a function of the parameter of interest, and an estimator of the Hessian matrix can be constructed by plugging in the initial estimator. In contrast, the SOS method estimates the Hessian matrix based on both the initial estimator and the subsample, which makes the asymptotic analysis more challenging. From the theoretical perspective, previous methods typically require that the initial estimator converges at a certain rate to guarantee the final estimator’s asymptotic distribution. In contrast, our Theorem \ref{theo: general} in Section \ref{sec: properties} establishes the asymptotic distribution of the SOS estimator without any restrictions on the convergence rate of the initial estimator, thereby offering greater flexibility in practice.

	\section{Asymptotic Properties}\label{sec: properties}
	
	In this section, we establish the asymptotic properties of the proposed estimator $\tilde{\btheta}_{\rm SOS}$. For this purpose, we consider the following regularity conditions.
	\begin{enumerate}[(C1)]
		\item There exist $G_{1}(\bz)$ and $G_{2}(\bz)$ such that $\|\nabla L(\bz;\btheta_{1})-\nabla L(\bz;\btheta_{2})\|\leq G_{1}(\bz)\|\btheta_{1}-\btheta_{2}\|$, $\|\nabla^2 L(\bz;\btheta_{1})-\nabla^2 L(\bz;\btheta_{2})\|\leq G_{2}(\bz)\|\btheta_{1}-\btheta_{2}\|$, and $\|\nabla^3 L(\bz;\btheta_{1})-\nabla^3 L(\bz;\btheta_{2})\|\leq G_{3}(\bz)\|\btheta_{1}-\btheta_{2}\|$ for any $\btheta_{1}$, $\btheta_{2}\in \Theta$ with $E[G_{k}(\bZ)^{2}] < \infty$ for $k=1,2,3$.
		\item There exists some $\tau>0$ such that $E\left\{\|\nabla L(\bZ_{i};\btheta_{0})\|^{2+\tau}\right\}<\infty$, $E\left\{
		\|\nabla^{2}L(\bZ_{i};\btheta_{0})\|^{2+\tau}\right\}<\infty$, and $E\left\{\|\nabla^{3}L(\bZ_{i};\btheta_{0})\|^2\right\}<\infty$.
		\item The minimum eigenvalue of $E\left\{\nabla^{2}L(\bZ;\btheta_{0})\right\}$ is bounded away from zero.
	\end{enumerate}

	\begin{theorem}\label{theo: general}
		Under conditions (C1)--(C3), as $n,N\to\infty$, we have
		\begin{equation*}
			\tilde{\btheta}_{\rm SOS}-\btheta_{0} = O_{P}(n^{-1}+N^{-1/2})
		\end{equation*}
		and
		\begin{equation}\label{eq: asy dist}
			\min(n, \sqrt{N})(\tilde{\btheta}_{\rm SOS}-\btheta_{0})\stackrel{d}{\to}\bg(\bU),
		\end{equation}
		where $\bU$ is a $2d+(d^2-d)/2+d$ dimensional random vector following from  $N(0,\bbV)$ with
		\begin{equation*}
			\bbV=\begin{pmatrix}
				\bbV_{11} & \bbV_{12} & \bbV_{13}\\
				\bbV_{21} & \bbV_{22} & \bbV_{23}\\
				\bbV_{31} & \bbV_{32} & \bbV_{33}
			\end{pmatrix}
			,
		\end{equation*}
		$\bbV_{11} =  E\left\{\nabla L(\bZ;\btheta_{0})\nabla L(\bZ;\btheta_{0})^{\T}\right\}$, $\bbV_{12} = \sqrt{r}\bbV_{11}$ with $r=\lim_{n,N\to \infty}n/N$, $\bbV_{13} = (1-r)E\left\{\nabla L(\bZ;\btheta_{0})\bV_{e}^{\T}\right\}$ with $\bV_{e}$ being the vectorized form of the upper triangle of $\nabla^{2}L(\bZ;\btheta_{0})$, $\bbV_{21} = \bbV_{12}^{\T}$, $\bbV_{22}=\bbV_{11},\bbV_{23} = (0)_{d\times((d^2-d)/2+d)}$, $\bbV_{31} = \bbV_{13}^{\T}$, $\bbV_{32} = \bbV_{23}^{\T}$, $\bbV_{33} = (1-r)E\left(\bV_{e}{\bV_{e}}^{\T}\right)$, and
		\begin{equation*}
			\bg(\bU) = c_{1}\bbH^{-1}(\bbM/2)\left\{(\bbH^{-1}\bU_{1})\otimes(\bbH^{-1}\bU_{1})\right\}-c_{2}\bbH^{-1}\bU_{2}-c_1\bbH^{-1}\bbU_C(\bbH^{-1}\bU_{1}),
		\end{equation*}
		$c_1 = \lim_{n,N\to \infty}\min(\sqrt{N},n)/n$, $c_{2} =  \lim_{n,N\to \infty}\min(\sqrt{N},n)/\sqrt{N}$, the notation $\otimes$ denotes Kronecker product,  $\bU_{1}$ consists of the first $d$ elements of $\bU$, $\bU_{2}$ consists of the $(d+1)$-th to the $(2d+1)$-th elements of $\bU$, and $\bbU_{C}$ is a $d\times d$ symmetrical matrix with upper triangle matrix consisted by the elements in $\bU_{3}$ arranged in rows with $\bU_{3}$ consisting of the $(2d+1)$-th to the $(2d+(d^2-d)/2+d)$-th elements of $\bU$, $\bbM=E\left\{\nabla^{3}L(\bZ;\btheta_{0})\right\}$ is a $d\times d^2$ matrix with the $j$-th row being the vectorized form of the Hessian matrix of $E\left\{\partial L(\bZ;\btheta_{0})/\partial\theta_{j}\right\}$, and
		$\bbH=E\left\{\nabla^{2}L(\bZ;\btheta_{0})\right\}$.
	\end{theorem}
	Theorem \ref{theo: general} shows that the convergence rate of $\tilde{\btheta}_{\rm SOS}$ is $\max\{n^{-1}, N^{-1/2}\}$, significantly faster than the existing subsampling estimators whose convergence rates are typically of order $n^{-1/2}$ \citep{Wang2018Subsampling, Wang2020Subsampling, Yu2020Subsampling}. In addition, Theorem \ref{theo: general} establishes the asymptotic distribution of $\tilde{\btheta}_{\rm SOS}$. 
	Although the SOS estimator is motivated by the asymptotic expansion under the assumption that $n/\sqrt{N} \to \infty$, its asymptotic distribution is established without requiring any condition on the ratio between $n$ and $\sqrt{N}$. In contrast, the one-step procedures in previous works such as \cite{Kutoyants2016OnMM}, \cite{Brouste2023Fast}, \cite{hariz2023fast}, and \cite{wang2024fitting} assume $n \gg \sqrt{N}$ to derive asymptotic distributions for their one-step estimators. However, when $N$ is extremely large and computational resources are limited, processing a subsample of size $n\gg \sqrt{N}$ can be impractical. 
	The multi-step procedure in \cite{Kutoyants2016OnMM} can alleviate the requirement but still imposes conditions on the ratio between $n$ and $\sqrt{N}$. For example, their two-step estimator requires $n \gg N^{1/4}$, which, while less restrictive, increases computational costs due to the additional step.
	Our results facilitate inference without order requirement on the subsample size $n$, providing a practical solution for situations with limited computational resources.

	The confidence interval for the proposed SOS estimator can be constructed based on the Monte Carlo approximation of the asymptotic distribution in Theorem \ref{theo: general} as follows:
	\begin{enumerate}[ Step 1.]
		\item Estimate $\bbV$ by $\widetilde{\bbV}$ which replaces the expectation in $\bbV$ by the corresponding sample form based on the subsample indexed by $S$ and replaces $\btheta_{0}$ by $\tilde{\btheta}_{\rm SOS}$;
		\item Generate $n$ random samples $\{\widetilde{\bU}_{1},\dots,\widetilde{\bU}_{n}\}$ from the normal distribution with mean $\bzero$ and variance $\widetilde{\bbV}$;
		\item Estimate $\bbM$ and $\bbH$ in Theorem \ref{theo: general} by $\widetilde{\bbM}$ and $\widetilde{\bbH}$, respectively, which replaces the expectation in $\bbM$ and $\bbH$ by their sample forms based on the subsample indexed by $S$ and replaces $\btheta_{0}$ in $\bbM$ and $\bbH$  by $\tilde{\btheta}_{\rm SOS}$;
		\item Estimate $\bg(\cdot)$ by $\tilde{\bg}(\cdot)$ which replaces $\bbM$ and $\bbH$ in $\bg(\cdot)$ by $\widetilde{\bbM}$ and $\widetilde{\bbH}$, respectively, and calculate $\tilde{\bg}(\widetilde{\bU}_{i})$ for $i=1,\dots,n$;
		\item Determine the lower and upper $\alpha/2$ quantiles of $\{\tilde{g}_{j}(\widetilde{\bU}_{i})\}_{i = 1}^{n}$, denoted by $\tilde{g}_{{\rm L},j}$ and $\tilde{g}_{{\rm U}, j}$ respectively, where $\tilde{g}_{j}(\widetilde{\bU}_{i})$ is the $j$-th component of $\tilde{\bg}(\widetilde{\bU}_{i})$ for $i = 1,\dots, n$ and $j=1,\dots,d$;
		\item For $j = 1,\dots, d$, construct the confidence interval for $\btheta_{0,j}$ with confidence level $1-\alpha$ as $[\tilde{\btheta}_{\rm SOS}-\tilde{g}_{{\rm U},j}/\min(\sqrt{N},n),\tilde{\btheta}_{\rm SOS}-\tilde{g}_{{\rm L},j}/\min(\sqrt{N},n)]$.
	\end{enumerate}
	See Section \ref{sec: simulation} for the numerical performance of the proposed confidence interval. 
	The asymptotic distribution in \eqref{eq: asy dist} is complicated and requires to be approximated by the Monte Carlo approximation. When the subsample size is not too small compared to the full data size in the sense that  $n/\sqrt{N}\to \infty$, $c_{1}=0$ and $c_{2}=1$ in Theorem \ref{theo: general} and the asymptotic distribution can be greatly simplified. In this case,
	$\tilde{\btheta}_{\rm SOS}$ is asymptotically equivalent to the full data-based estimator $\hat{\btheta}_{\rm full}$. We summarize the result in the following corollary.
	
	\begin{corollary}\label{corr: special}
		Under conditions (C1)--(C3), as $n,N\to\infty$ and $n/\sqrt{N}\to \infty$, we have
		\begin{equation*}
			\sqrt{N}(\tilde{\btheta}_{\rm SOS}-\hat{\btheta}_{\rm full})\stackrel{p}{\to}0
		\end{equation*}
		and
		\begin{equation*}
			\sqrt{N}(\tilde{\btheta}_{\rm SOS}-\btheta_{0})\stackrel{d}{\to}N(0,\bbV_{c}),
		\end{equation*}
		where $\bbV_{c}=\bbH^{-1}\bbV_{11}\bbH^{-1}$ with $\bbH$ and $\bbV_{11}$ being defined in Theorem \ref{theo: general}.
	\end{corollary}

	The asymptotic result in Corollary \ref{corr: special} is analogous to those in \citet{Kutoyants2016OnMM},  \cite{Brouste2023Fast}, and \citet{hariz2023fast} but applies to a different problem. Specifically, they focus on specific parametric models with possibly correlated data. In contrast, Corollary \ref{corr: special} is established under a general M-estimation framework with i.i.d. data.

	\section{Numerical Experiments}\label{sec: simulation}
	
	In this section, we conduct numerical experiments to evaluate the finite sample performance of the proposed method. For comparison, we also calculate several widely studied subsampling estimators in existing literature, including the uniform subsampling estimator (UNI) defined in \eqref{eq:uni est}, the inverse probability weighted subsampling estimator (IPW) proposed by \cite{Wang2018Subsampling} which employs the optimal nonuniform subsampling probability (\emph{opt}NSP) that minimizes the asymptotic mean squared error of the IPW estimator, the \emph{opt}NSP-based maximum sampled conditional likelihood subsampling estimator (MSCL) proposed by \cite{Wang2020LikeliEfficient}, and the \emph{opt}NSP-based empirical likelihood weighting subsampling estimator (ELW) proposed by \cite{fan2022nearly}. As a reference, we also calculate the full data-based estimator (FULL) defined in \eqref{eq:full est}.
	
	In the simulation,  we generate a set of i.i.d. observations $\bZ_i = (\bX_i, Y_i)$, $i=1,\dots,N$, where $\bX_{i}=(X_{i1},\dots,X_{i9})$ is a $9$-dimensional covariate vector, and $X_{ij},j=1,\dots,9$ are independently generated from $U(-1,1)$. 
	We consider two regression models: logistic regression and Weibull regression. In the logistic regression model, the response variable $Y_i$ follows a Bernoulli distribution with mean  $\exp(\gamma_{0}+\bX_{i}^{\T}\bbeta_{0})/\{1+\exp(\gamma_{0}+\bX_{i}^{\T}\bbeta_{0})\}$, where $\gamma_{0}=0$ and $\bbeta_{0}$ is a $9$-dimensional vector with all elements set to $0.2$. The parameter of interest is $\btheta_{0}=(\gamma_{0},\bbeta_{0}^{\T})^{\T}$. In the Weibull regression model, the response variable $Y_i$ is generated as $Y_i = W_{i}\exp(\gamma_{0}+\bX_{i}^{\T}\bbeta_{0}/\alpha_{0})$, where $W_{i}$ follows a Weibull distribution with a shape parameter $\alpha_{0}=0.5$ and scale parameter $1$. The parameter of interest is $\btheta_{0}=(\alpha_{0},\gamma_{0},\bbeta_{0}^{\T})^{\T}$.
	For each of the models, we fix the full sample size $N=10^{6}$ and vary the subsample size $n$ to be $5\times 10^{3}$, $10^{4}$, $2\times10^{4}$ and $5\times10^{4}$. In addition, to calculate the \emph{opt}NSP, we generate a pilot sample of size $200$.
	To evaluate the performance of an estimator $\bar{\btheta}$ , we calculate its empirical
	${\rm RMSE}=\left(K^{-1}\sum_{k=1}^{K}\|\bar{\btheta}^{k}-\btheta_{0}\|^2\right)^{1/2}$
	based on $K=1000$ repetitions, where $\bar{\btheta}^{k}$ is the estimate from the $k$-th run.
	Figure \ref{fig: RMSE vary sub} plots the RMSE of the above estimators under the logistic and Weibull models for different subsample sizes. Due to computational complexity, we do not provide results for the MSCL estimator under the Weibull model, as the integral operation involved in this estimator is difficult to express in closed form.
	
	\begin{figure}[htpb]
		\centering
		\includegraphics[scale=0.35]{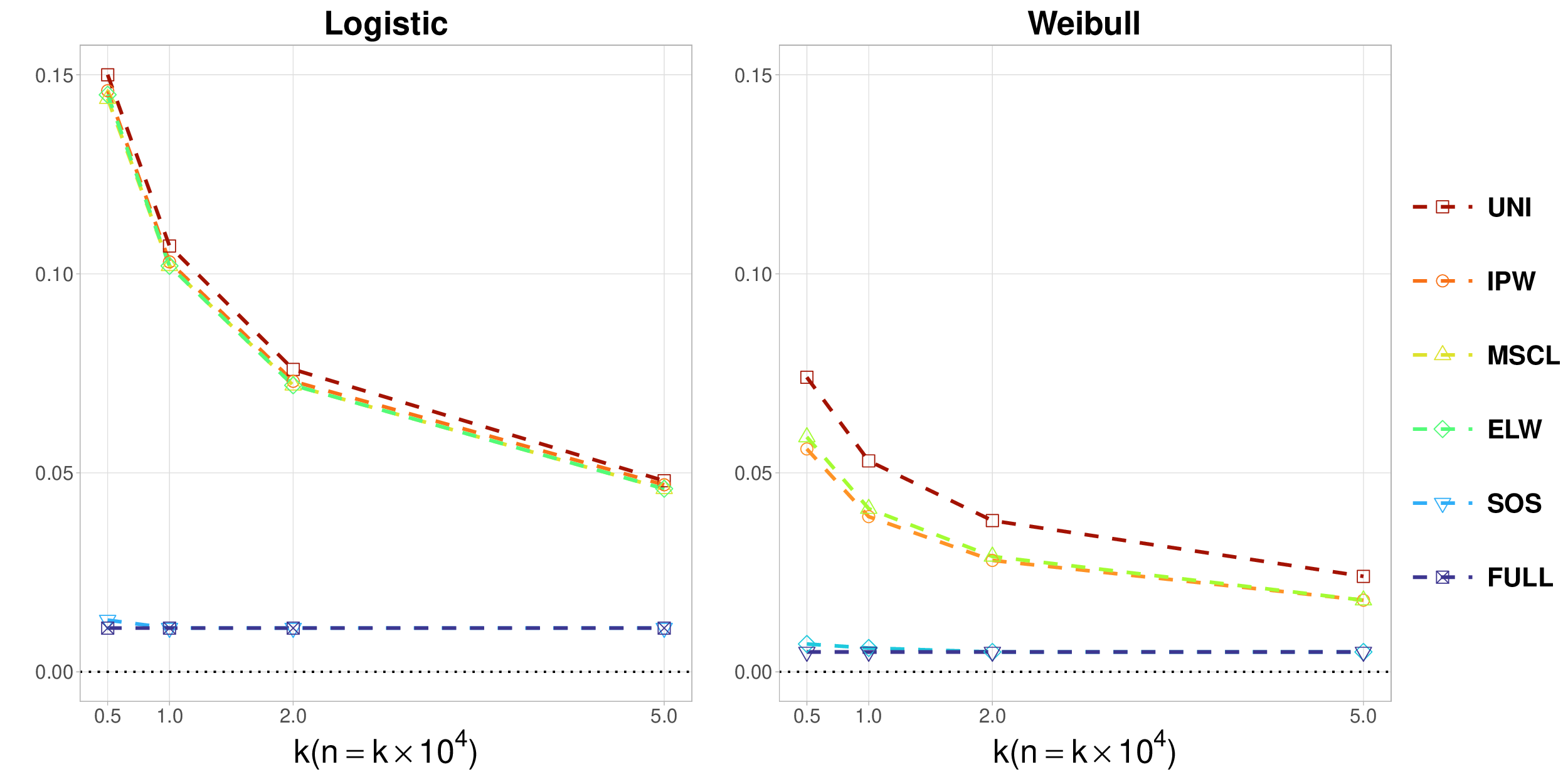}
		\caption{RMSE of different estimators under the logistic and Weibull models with $N=10^{6}$ and $n=k\times10^{4}$.}
		\label{fig: RMSE vary sub}
	\end{figure}
	
	Figure \ref{fig: RMSE vary sub} illustrates that UNI exhibits a large RMSE, indicating lower estimation efficiency. The IPW, MSCL, and ELW estimators under the logistic model, as well as the IPW and ELW estimators under the Weibull model, achieve lower RMSE than UNI, with a particularly notable improvement in estimation efficiency under the Weibull model.
	Nevertheless, despite these improvements, existing subsampling estimators still suffer from substantial estimation efficiency loss compared to the full data-based estimator.
	In contrast, the proposed SOS estimator outperforms other subsampling estimators in terms of RMSE and even performs similarly to the full data-based estimator under both models.
	This phenomenon aligns with our theoretical results in Theorem \ref{theo: general} and Corollary \ref{corr: special}.
	The SOS estimator achieves a faster convergence rate than existing subsampling estimators and even has the same estimation efficiency as the full data-based estimator when the subsample size is not too small compared to the full data size.
	
	To further evaluate these estimators, we present the bias and the Monte Carlo standard deviation (SD) of the above estimators for each dimension under the logistic and Weibull models in Tables \ref{tab: simul logistic vary sub} and \ref{tab: simul Weibull vary sub}, respectively.
	The bias and SD of the full data-based estimator do not change as the subsample size varies and hence we present the results here.
	The bias and SD of the full data-based estimator under the logistic model are (0.000, -0.003,  0.000, -0.001,  0.000,  0.001,  0.001,  0.000, -0.001,  0.002) and
	(0.020, 0.035, 0.035, 0.036, 0.035, 0.037, 0.034, 0.036, 0.036, 0.036), respectively.
	The bias and SD of the full data-based estimator under the Weibull model are (0.000,  0.000,  0.001,  0.000,  0.000,  0.000, -0.001, -0.001,  0.000, -0.001, -0.001) and (0.004, 0.010, 0.017, 0.017, 0.018, 0.017, 0.018, 0.017, 0.017, 0.018, 0.017), respectively.

	\begin{table}[htbp]
		\begin{center}
			\caption{Bias and SD of different estimators for the $j$-th dimension of $\btheta_{0}$ under the logistic model with $N = 10^{6}$ and $n=k\times 10^{4}$. The numbers in the table are the simulation results multiplied by $10$}{	
				\resizebox{.8\columnwidth}{!}{
					{\begin{tabular}{ccccccccccccccc}
							\toprule
							\multirow{2}{15pt}{$k$} &\multirow{2}{15pt}{{$j$}} 
							&\multicolumn{2}{c}{UNI}& \multicolumn{2}{c}{IPW}& \multicolumn{2}{c}{MSCL}&\multicolumn{2}{c}{ELW}&\multicolumn{2}{c}{SOS}\\
							\cmidrule(lr){3-4}  \cmidrule(lr){5-6} \cmidrule(lr){7-8}\cmidrule(lr){9-10}\cmidrule(lr){11-12}
							& & Bias &SD  &Bias &SD  &Bias &SD  &Bias &SD&Bias &SD\\
							\midrule
							\multirow{10}{15pt}{0.5}	
							&1   &0.009 &0.292 &-0.008 &0.276 &-0.010 &0.273  &0.001 &0.057  &0.000 &0.025\\
							&2   &0.030 &0.489 &-0.019 &0.477 &-0.015 &0.466 &-0.017 &0.478 &-0.006 &0.043\\
							&3   &0.011 &0.489 &-0.016 &0.500 &-0.010 &0.494 &-0.002 &0.485 &-0.005 &0.043\\
							&4   &0.026 &0.481 &-0.012 &0.474 &-0.006 &0.468  &0.019 &0.475 &-0.006 &0.041\\
							&5   &0.008 &0.498 &-0.007 &0.477 &-0.006 &0.471  &0.005 &0.484 &-0.006 &0.042\\
							&6  &-0.005 &0.478  &0.013 &0.488  &0.017 &0.479  &0.017 &0.499 &-0.003 &0.044\\
							&7   &0.005 &0.504 &-0.011 &0.475 &-0.010 &0.470  &0.006 &0.504 &-0.003 &0.041\\
							&8  &-0.001 &0.505 &-0.011 &0.451 &-0.009 &0.444  &0.005 &0.472 &-0.005 &0.044\\
							&9   &0.028 &0.480 &-0.031 &0.479 &-0.032 &0.474 &-0.007 &0.465 &-0.005 &0.042\\
							&10  &0.019 &0.500  &0.002 &0.468 &-0.002 &0.463  &0.037 &0.479 &-0.003 &0.044\\
							\vspace{-5pt}
							\multirow{12}{15pt}{1} \\		
							&1   &0.003 &0.204 &-0.004 &0.192 &-0.004 &0.189  &0.001 &0.043  &0.000 &0.021\\
							&2   &0.011 &0.345 &-0.001 &0.339 &-0.002 &0.332 &-0.008 &0.333 &-0.005 &0.037\\
							&3   &0.002 &0.345 &-0.014 &0.346 &-0.009 &0.339 &-0.009 &0.344 &-0.002 &0.037\\
							&4   &0.011 &0.355 &-0.002 &0.339  &0.000 &0.337  &0.015 &0.348 &-0.004 &0.037\\
							&5   &0.003 &0.344 &-0.010 &0.336 &-0.011 &0.333 &-0.001 &0.335 &-0.003 &0.037\\
							&6   &0.013 &0.356 &-0.003 &0.347 &-0.001 &0.343  &0.004 &0.338 &-0.001 &0.039\\
							&7  &-0.003 &0.350 &-0.004 &0.324 &-0.006 &0.322  &0.007 &0.349 &-0.001 &0.036\\
							&8   &0.007 &0.360 &-0.007 &0.341 &-0.005 &0.336 &-0.008 &0.341 &-0.002 &0.038\\
							&9   &0.011 &0.341 &-0.013 &0.342 &-0.014 &0.335  &0.001 &0.348 &-0.004 &0.038\\
							&10  &0.008 &0.359 &-0.013 &0.331 &-0.015 &0.325  &0.017 &0.331  &0.000 &0.039\\
							\vspace{-5pt}
							\multirow{12}{15pt}{2} \\		
							&1  &-0.002 &0.147  &0.001 &0.142  &0.002 &0.140  &0.000 &0.032  &0.000 &0.020\\
							&2  &-0.008 &0.252  &0.005 &0.231  &0.005 &0.227 &-0.011 &0.235 &-0.004 &0.036\\
							&3   &0.005 &0.243 &-0.006 &0.241 &-0.004 &0.238  &0.006 &0.233 &-0.001 &0.035\\
							&4   &0.009 &0.250 &-0.007 &0.241 &-0.005 &0.237  &0.011 &0.244 &-0.002 &0.036\\
							&5   &0.003 &0.246  &0.001 &0.240  &0.001 &0.239 &-0.010 &0.239 &-0.001 &0.035\\
							&6   &0.007 &0.254  &0.008 &0.236  &0.009 &0.233  &0.004 &0.233  &0.000 &0.038\\
							&7  &-0.007 &0.253  &0.009 &0.247  &0.008 &0.242  &0.002 &0.252  &0.000 &0.035\\
							&8  &-0.009 &0.250 &-0.009 &0.242 &-0.006 &0.239  &0.000 &0.236 &-0.001 &0.036\\
							&9  &-0.001 &0.247 &-0.008 &0.240 &-0.008 &0.236  &0.010 &0.239 &-0.002 &0.037\\
							&10  &0.004 &0.251 &-0.021 &0.234 &-0.021 &0.230  &0.021 &0.241  &0.001 &0.037\\
							\vspace{-5pt}
							\multirow{12}{15pt}{5} \\		
							&1   &0.001 &0.090  &0.002 &0.089  &0.002 &0.088  &0.000 &0.025  &0.000 &0.020\\
							&2  &-0.006 &0.160 &-0.003 &0.150 &-0.003 &0.148  &0.007 &0.149 &-0.003 &0.035\\
							&3   &0.006 &0.148  &0.003 &0.156  &0.004 &0.154  &0.008 &0.150  &0.000 &0.035\\
							&4  &-0.001 &0.162 &-0.007 &0.154 &-0.006 &0.152  &0.000 &0.146 &-0.002 &0.036\\
							&5  &-0.003 &0.160 &-0.001 &0.151 &-0.001 &0.150 &-0.006 &0.155 &-0.001 &0.035\\
							&6   &0.000 &0.161  &0.004 &0.146  &0.004 &0.144  &0.001 &0.150  &0.000 &0.037\\
							&7  &-0.002 &0.162  &0.001 &0.152  &0.001 &0.149  &0.002 &0.156  &0.000 &0.034\\
							&8   &0.002 &0.156  &0.000 &0.155  &0.000 &0.153 &-0.001 &0.154 &-0.001 &0.036\\
							&9  &-0.001 &0.155 &-0.006 &0.150 &-0.008 &0.149  &0.003 &0.151 &-0.001 &0.036\\
							&10  &0.006 &0.160 &-0.013 &0.156 &-0.012 &0.154  &0.009 &0.154  &0.001 &0.036\\
							\bottomrule
			\end{tabular}}}}
		\label{tab: simul logistic vary sub}
	\end{center}
\end{table}

\begin{table}[htbp]
\begin{center}
	\caption{Bias and SD of different estimators for the $j$-th dimension of $\btheta_{0}$ under the Weibull model with $N = 10^{6}$ and $n=k\times 10^{4}$. The numbers in the table are the simulation results multiplied by $10$}{	
		\resizebox{.7\columnwidth}{!}{
			{\begin{tabular}{ccccccccccccccc}
					\toprule
					\multirow{2}{15pt}{$k$} &\multirow{2}{15pt}{{$j$}} 
					&\multicolumn{2}{c}{UNI}& \multicolumn{2}{c}{IPW}&\multicolumn{2}{c}{ELW}&\multicolumn{2}{c}{SOS}\\
					\cmidrule(lr){3-4}  \cmidrule(lr){5-6} \cmidrule(lr){7-8}\cmidrule(lr){9-10}
					& & Bias &SD  &Bias &SD  &Bias &SD &Bias &SD\\
					\midrule
					\multirow{11}{15pt}{0.5}	
					&1   &0.007 &0.055  &0.000 &0.057  &0.002 &0.041 &-0.007 &0.005\\
					&2  &-0.001 &0.153  &0.000 &0.123 &-0.001 &0.111 &-0.004 &0.012\\
					&3  &-0.003 &0.241  &0.008 &0.180  &0.011 &0.186 &-0.002 &0.024\\
					&4   &0.003 &0.244  &0.003 &0.184  &0.002 &0.190 &-0.002 &0.024\\
					&5   &0.001 &0.237  &0.007 &0.180 &-0.007 &0.194 &-0.003 &0.024\\
					&6   &0.005 &0.248  &0.005 &0.184 &-0.004 &0.191 &-0.002 &0.023\\
					&7  &-0.006 &0.239 &-0.001 &0.185  &0.000 &0.192 &-0.003 &0.024\\
					&8   &0.004 &0.237 &-0.006 &0.181  &0.004 &0.195 &-0.004 &0.023\\
					&9   &0.012 &0.245 &-0.005 &0.179  &0.006 &0.197 &-0.002 &0.025\\
					&10 &-0.005 &0.243  &0.008 &0.185  &0.005 &0.193 &-0.002 &0.023\\
					&11  &0.007 &0.238  &0.008 &0.173 &-0.010 &0.188 &-0.003 &0.024\\
					\vspace{-5pt}
					\multirow{12}{15pt}{1} \\		
					&1   &0.003 &0.040 &-0.001 &0.040  &0.001 &0.029 &-0.004 &0.004\\
					&2   &0.000 &0.106  &0.002 &0.088  &0.000 &0.078 &-0.002 &0.011\\
					&3  &-0.002 &0.173  &0.005 &0.126  &0.004 &0.131  &0.000 &0.019\\
					&4   &0.002 &0.175  &0.003 &0.130  &0.003 &0.134 &-0.002 &0.019\\
					&5   &0.000 &0.175 &-0.002 &0.128 &-0.008 &0.133 &-0.001 &0.020\\
					&6   &0.002 &0.172  &0.000 &0.129  &0.000 &0.131 &-0.002 &0.018\\
					&7  &-0.010 &0.173 &-0.002 &0.127 &-0.004 &0.138 &-0.002 &0.019\\
					&8   &0.001 &0.166  &0.000 &0.125  &0.005 &0.134 &-0.002 &0.019\\
					&9   &0.005 &0.177  &0.000 &0.128  &0.003 &0.135 &-0.001 &0.019\\
					&10 &-0.001 &0.176  &0.008 &0.130  &0.000 &0.131 &-0.002 &0.019\\
					&11  &0.004 &0.175  &0.012 &0.122 &-0.006 &0.130 &-0.002 &0.020\\
					\vspace{-5pt}
					\multirow{12}{15pt}{2} \\		
					&1   &0.001 &0.028 &-0.001 &0.029  &0.000 &0.020 &-0.002 &0.004\\
					&2   &0.002 &0.077  &0.002 &0.063  &0.001 &0.053 &-0.001 &0.011\\
					&3   &0.008 &0.124  &0.003 &0.087  &0.002 &0.094  &0.000 &0.018\\
					&4  &-0.001 &0.121  &0.002 &0.093 &-0.001 &0.096 &-0.001 &0.018\\
					&5  &-0.004 &0.125  &0.000 &0.091 &-0.005 &0.094  &0.000 &0.018\\
					&6   &0.001 &0.123  &0.000 &0.091 &-0.001 &0.093 &-0.001 &0.017\\
					&7  &-0.010 &0.127 &-0.003 &0.087 &-0.004 &0.096 &-0.001 &0.018\\
					&8   &0.001 &0.123 &-0.002 &0.091  &0.005 &0.096 &-0.001 &0.018\\
					&9   &0.010 &0.120  &0.001 &0.094  &0.003 &0.092  &0.000 &0.018\\
					&10  &0.001 &0.124  &0.003 &0.091 &-0.003 &0.095 &-0.001 &0.018\\
					&11  &0.001 &0.125  &0.007 &0.089 &-0.001 &0.095 &-0.001 &0.018\\
					\vspace{-5pt}
					\multirow{12}{15pt}{5} \\		
					&1   &0.000 &0.017  &0.000 &0.018  &0.000 &0.013 &-0.001 &0.004\\
					&2   &0.000 &0.047  &0.000 &0.041  &0.000 &0.034 &-0.001 &0.010\\
					&3   &0.004 &0.079 &-0.001 &0.055  &0.000 &0.061  &0.001 &0.017\\
					&4  &-0.002 &0.076  &0.004 &0.058 &-0.004 &0.061  &0.000 &0.017\\
					&5   &0.000 &0.080 &-0.001 &0.059  &0.002 &0.060  &0.000 &0.018\\
					&6   &0.002 &0.077  &0.000 &0.055  &0.002 &0.058 &-0.001 &0.017\\
					&7  &-0.005 &0.082 &-0.002 &0.057  &0.000 &0.061 &-0.001 &0.018\\
					&8   &0.002 &0.080 &-0.002 &0.056 &-0.002 &0.059 &-0.001 &0.017\\
					&9   &0.004 &0.077  &0.000 &0.057 &-0.001 &0.057  &0.000 &0.017\\
					&10 &-0.002 &0.077  &0.000 &0.058  &0.001 &0.059 &-0.001 &0.018\\
					&11 &-0.001 &0.078  &0.002 &0.057  &0.003 &0.062 &-0.001 &0.017\\
					\bottomrule
	\end{tabular}}}}
\label{tab: simul Weibull vary sub}
\end{center}
\end{table}

	Table \ref{tab: simul logistic vary sub} shows that all subsampling estimators under the logistic model exhibit slight bias, and the bias is similar to that of the full data-based estimator.
	However, the SD of the UNI, IPW, MSCL, and ELW estimators is considerably larger than that of the full data-based estimator. In contrast, the proposed SOS estimator demonstrates smaller SDs than other subsampling estimators for all subsample sizes and approaches the SD of the full data-based estimator when the subsample size $n$ is not too small.
	These findings provide empirical evidence for the asymptotic equivalence between the proposed estimator and the full data-based estimator, as established in Corollary \ref{corr: special}.
	Moreover, we observe that the SD of the IPW, MSCL, and ELW estimators can be larger than that of the UNI subsampling estimator for certain components of the parameter, e.g., the third component under the logistic model. This is because the IPW, MSCL, and ELW estimators are calculated based on nonuniform subsampling probabilities that minimize the trace of the resulting estimators' asymptotic variances. Consequently, improvement for each component of the parameter cannot be guaranteed. In contrast, the SOS estimator consistently exhibits significantly smaller SD than the UNI subsampling estimator for each component of the parameter of interest.
	Similar observations can be made under the Weibull model, as shown in Table \ref{tab: simul Weibull vary sub}.

	Furthermore, we investigate the effect of the full sample size $N$ on the performance of the proposed estimator. We fix the subsample size $n=5\times 10^{4}$ and vary the full sample size $N$ to be $10^{6}$, $2\times10^{6}$, $5\times10^{6}$ and $10^{7}$.
	Figure \ref{fig: RMSE vary full} plots the RMSE of the above estimators under the logistic and Weibull models for different full data sizes.

	\begin{figure}[htpb]
		\centering
		\includegraphics[scale=0.35]{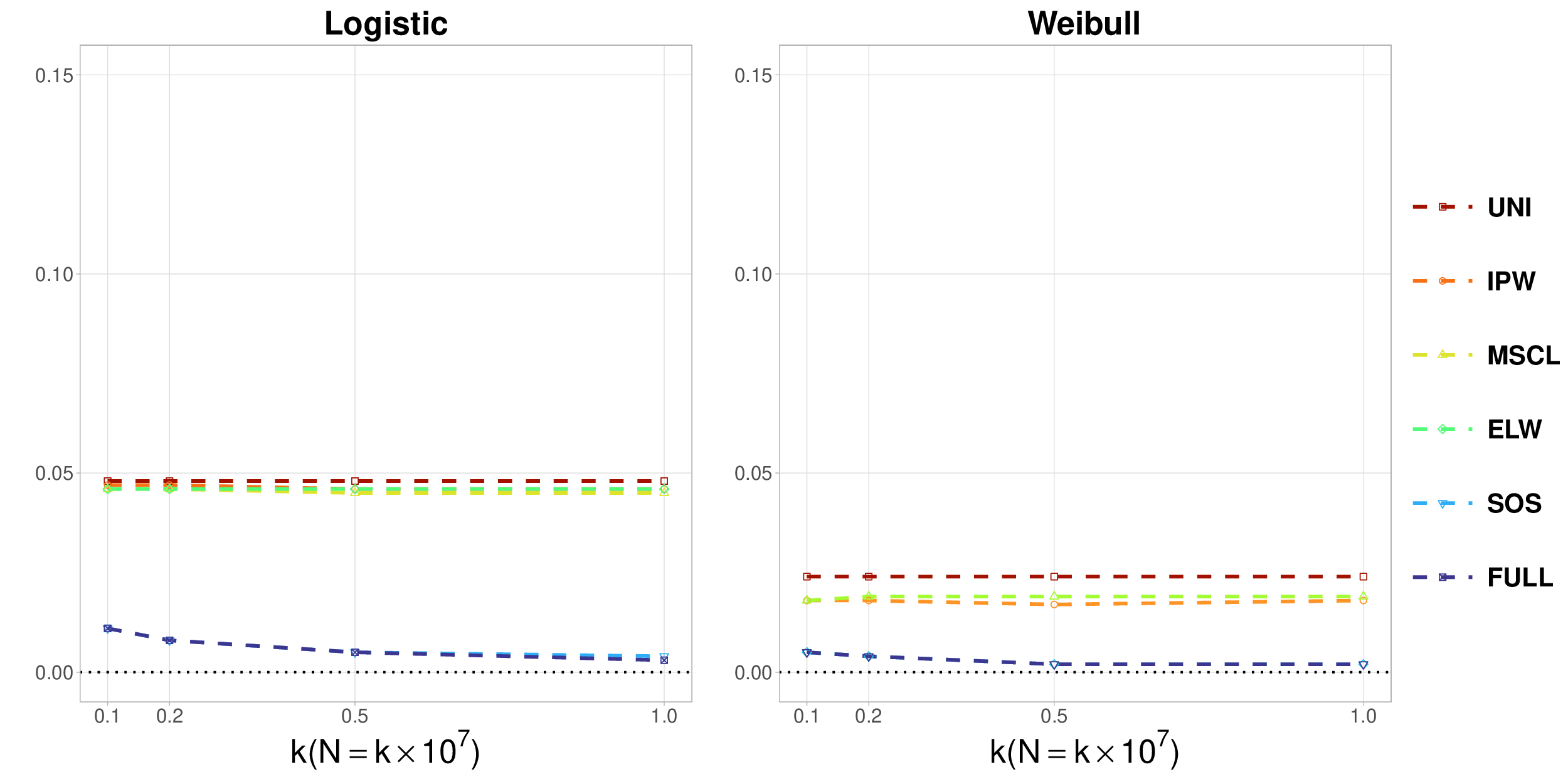}
		\caption{RMSE of different estimators under the logistic and Weibull models with $N=k\times 10^{7}$ and $n=5\times10^{4}$.}
		\label{fig: RMSE vary full}
	\end{figure}
	
	Figure \ref{fig: RMSE vary full} demonstrates that the proposed SOS estimator performs comparably to the full data-based estimator, with the RMSE decreasing as the full data size increases. On the other hand, the RMSE of the UNI, IPW, MSCL, and ELW estimators shows little change as the full data size increases and remains significantly larger than the RMSE of the SOS estimator. 
	
	We also present the bias and SD of the above estimators for different full data sizes under the logistic and Weibull models in Tables \ref{tab: simul logistic vary full} and \ref{tab: simul weibull vary full}, respectively.

	\begin{table}[htbp]
		\begin{center}
			\caption{Bias and SD of different estimators for the $j$-th dimension of $\btheta_{0}$ under the logistic model with $N = k\times 10^{7}$ and $n=5\times 10^{4}$. The numbers in the table are the simulation results multiplied by $10$}{	
				\resizebox{1\columnwidth}{!}{
					{\begin{tabular}{ccccccccccccccc}
							\toprule
							\multirow{2}{15pt}{$k$} &\multirow{2}{15pt}{{$j$}} 
							&\multicolumn{2}{c}{FULL}&\multicolumn{2}{c}{UNI}& \multicolumn{2}{c}{IPW}& \multicolumn{2}{c}{MSCL}&\multicolumn{2}{c}{ELW}&\multicolumn{2}{c}{SOS}\\
							\cmidrule(lr){3-4}  \cmidrule(lr){5-6} \cmidrule(lr){7-8}\cmidrule(lr){9-10}\cmidrule(lr){11-12}\cmidrule(lr){13-14}
							& & Bias &SD& Bias &SD  &Bias &SD  &Bias &SD  &Bias &SD&Bias &SD\\
							\midrule
							\multirow{10}{15pt}{0.1}	
							&1   &0.000 &0.020  &0.001 &0.090  &0.002 &0.089  &0.002 &0.088  &0.000 &0.025  &0.000 &0.020\\
							&2  &-0.003 &0.035 &-0.006 &0.160 &-0.003 &0.150 &-0.003 &0.148  &0.007 &0.149 &-0.003 &0.035\\
							&3   &0.000 &0.035  &0.006 &0.148  &0.003 &0.156  &0.004 &0.154  &0.008 &0.150  &0.000 &0.035\\
							&4  &-0.001 &0.036 &-0.001 &0.162 &-0.007 &0.154 &-0.006 &0.152  &0.000 &0.146 &-0.002 &0.036\\
							&5   &0.000 &0.035 &-0.003 &0.160 &-0.001 &0.151 &-0.001 &0.150 &-0.006 &0.155 &-0.001 &0.035\\
							&6   &0.001 &0.037  &0.000 &0.161  &0.004 &0.146  &0.004 &0.144  &0.001 &0.150  &0.000 &0.037\\
							&7   &0.001 &0.034 &-0.002 &0.162  &0.001 &0.152  &0.001 &0.149  &0.002 &0.156  &0.000 &0.034\\
							&8   &0.000 &0.036  &0.002 &0.156  &0.000 &0.155  &0.000 &0.153 &-0.001 &0.154 &-0.001 &0.036\\
							&9  &-0.001 &0.036 &-0.001 &0.155 &-0.006 &0.150 &-0.008 &0.149  &0.003 &0.151 &-0.001 &0.036\\
							&10  &0.002 &0.036  &0.006 &0.160 &-0.013 &0.156 &-0.012 &0.154  &0.009 &0.154  &0.001 &0.036\\
							\vspace{-5pt}
							\multirow{12}{15pt}{0.2} \\		
							&1   &0.000 &0.014  &0.002 &0.089  &0.004 &0.089  &0.003 &0.087  &0.000 &0.021  &0.000 &0.014\\
							&2   &0.000 &0.024  &0.006 &0.160 &-0.006 &0.162 &-0.005 &0.159 &-0.001 &0.153 &-0.001 &0.024\\
							&3  &-0.001 &0.025  &0.003 &0.154  &0.001 &0.152  &0.001 &0.149  &0.001 &0.151 &-0.001 &0.025\\
							&4   &0.000 &0.025  &0.001 &0.163  &0.006 &0.153  &0.005 &0.151 &-0.006 &0.156 &-0.001 &0.025\\
							&5  &-0.001 &0.024  &0.004 &0.158  &0.000 &0.153 &-0.001 &0.151 &-0.010 &0.153 &-0.001 &0.024\\
							&6   &0.000 &0.025  &0.002 &0.154  &0.006 &0.152  &0.005 &0.151 &-0.002 &0.152  &0.000 &0.025\\
							&7   &0.000 &0.024  &0.000 &0.162  &0.002 &0.152  &0.004 &0.151 &-0.011 &0.151 &-0.001 &0.024\\
							&8   &0.000 &0.025  &0.003 &0.165 &-0.003 &0.151 &-0.003 &0.149  &0.000 &0.157 &-0.001 &0.025\\
							&9   &0.000 &0.025 &-0.004 &0.157 &-0.002 &0.153 &-0.002 &0.150  &0.000 &0.155  &0.000 &0.025\\
							&10  &0.000 &0.025  &0.001 &0.151  &0.010 &0.146  &0.010 &0.144 &-0.004 &0.150 &-0.001 &0.025\\
							\vspace{-5pt}
							\multirow{12}{15pt}{0.5} \\		
							&1  &0.000 &0.009  &0.000 &0.091  &0.000 &0.091  &0.001 &0.089  &0.000 &0.018  &0.000 &0.009\\
							&2  &0.001 &0.015  &0.007 &0.159 &-0.003 &0.148 &-0.003 &0.147 &-0.001 &0.151  &0.001 &0.015\\
							&3  &0.001 &0.016  &0.005 &0.156 &-0.001 &0.150 &-0.001 &0.147 &-0.002 &0.158  &0.000 &0.016\\
							&4  &0.000 &0.016  &0.001 &0.160  &0.004 &0.155  &0.005 &0.152  &0.001 &0.155  &0.000 &0.016\\
							&5  &0.000 &0.015  &0.005 &0.155  &0.002 &0.147  &0.001 &0.146 &-0.002 &0.156 &-0.001 &0.016\\
							&6  &0.000 &0.016 &-0.001 &0.158 &-0.001 &0.149  &0.000 &0.147  &0.002 &0.153  &0.000 &0.016\\
							&7  &0.000 &0.017  &0.000 &0.158  &0.002 &0.143  &0.002 &0.142  &0.007 &0.153 &-0.001 &0.017\\
							&8  &0.000 &0.016 &-0.010 &0.158  &0.005 &0.145  &0.006 &0.142  &0.007 &0.149 &-0.001 &0.016\\
							&9  &0.001 &0.016  &0.008 &0.154  &0.003 &0.152  &0.002 &0.150  &0.002 &0.153  &0.000 &0.016\\
							&10 &0.000 &0.015  &0.003 &0.158  &0.008 &0.155  &0.008 &0.153  &0.005 &0.154  &0.000 &0.016\\
							\vspace{-5pt}
							\multirow{12}{15pt}{1} \\		
							&1  &-0.001 &0.006  &0.002 &0.093  &0.001 &0.088  &0.000 &0.086  &0.000 &0.017  &0.000 &0.007\\
							&2   &0.001 &0.010  &0.001 &0.157  &0.005 &0.143  &0.005 &0.141 &-0.002 &0.151  &0.000 &0.011\\
							&3   &0.002 &0.011  &0.003 &0.159 &-0.003 &0.153 &-0.004 &0.149 &-0.002 &0.158  &0.000 &0.012\\
							&4  &-0.001 &0.012  &0.001 &0.155 &-0.004 &0.151 &-0.003 &0.147 &-0.006 &0.156  &0.000 &0.012\\
							&5   &0.000 &0.011  &0.011 &0.156 &-0.004 &0.155 &-0.003 &0.154  &0.004 &0.157 &-0.001 &0.012\\
							&6   &0.000 &0.011 &-0.002 &0.160 &-0.005 &0.145 &-0.005 &0.144  &0.000 &0.150 &-0.001 &0.011\\
							&7  &-0.001 &0.011  &0.009 &0.154  &0.002 &0.148  &0.003 &0.147  &0.005 &0.155  &0.000 &0.012\\
							&8   &0.001 &0.011  &0.003 &0.159 &-0.004 &0.153 &-0.002 &0.150 &-0.006 &0.159  &0.000 &0.011\\
							&9   &0.001 &0.010  &0.006 &0.157  &0.004 &0.147  &0.004 &0.145  &0.000 &0.149  &0.000 &0.011\\
							&10  &0.000 &0.010  &0.003 &0.162 &-0.004 &0.151 &-0.004 &0.149  &0.001 &0.155  &0.000 &0.011\\
							\bottomrule
			\end{tabular}}}}
		\label{tab: simul logistic vary full}
	\end{center}
\end{table}

\begin{table}[htbp]
\begin{center}
	\caption{Bias and SD of different estimators for the $j$-th dimension of $\btheta_{0}$ under the Weibull model with $N = k\times 10^{7}$ and $n=5\times 10^{4}$. The numbers in the table are the simulation results multiplied by $10$}{	
		\resizebox{.9\columnwidth}{!}{
			{\begin{tabular}{ccccccccccccccc}
					\toprule
					\multirow{2}{15pt}{$k$} &\multirow{2}{15pt}{{$j$}} 
					&\multicolumn{2}{c}{FULL}&\multicolumn{2}{c}{UNI}& \multicolumn{2}{c}{IPW}&\multicolumn{2}{c}{ELW}&\multicolumn{2}{c}{SOS}\\
					\cmidrule(lr){3-4}  \cmidrule(lr){5-6} \cmidrule(lr){7-8}\cmidrule(lr){9-10}\cmidrule(lr){11-12}
					& & {Bias} &{SD}  &{Bias} &{SD}  &{Bias} &{SD} &{Bias}&{SD}  &Bias &SD\\
					\midrule
					\multirow{11}{15pt}{0.1}	
					&1   &0.000 &0.004  &0.000 &0.017  &0.000 &0.018  &0.000 &0.013 &-0.001 &0.004\\
					&2   &0.000 &0.010  &0.000 &0.047  &0.000 &0.041  &0.000 &0.034 &-0.001 &0.010\\
					&3   &0.001 &0.017  &0.004 &0.079 &-0.001 &0.055  &0.000 &0.061  &0.001 &0.017\\
					&4   &0.000 &0.017 &-0.002 &0.076  &0.004 &0.058 &-0.004 &0.061  &0.000 &0.017\\
					&5   &0.000 &0.018  &0.000 &0.080 &-0.001 &0.059  &0.002 &0.060  &0.000 &0.018\\
					&6   &0.000 &0.017  &0.002 &0.077  &0.000 &0.055  &0.002 &0.058 &-0.001 &0.017\\
					&7  &-0.001 &0.018 &-0.005 &0.082 &-0.002 &0.057  &0.000 &0.061 &-0.001 &0.018\\
					&8  &-0.001 &0.017  &0.002 &0.080 &-0.002 &0.056 &-0.002 &0.059 &-0.001 &0.017\\
					&9   &0.000 &0.017  &0.004 &0.077  &0.000 &0.057 &-0.001 &0.057  &0.000 &0.017\\
					&10 &-0.001 &0.018 &-0.002 &0.077  &0.000 &0.058  &0.001 &0.059 &-0.001 &0.018\\
					&11 &-0.001 &0.017 &-0.001 &0.078  &0.002 &0.057  &0.003 &0.062 &-0.001 &0.017\\
					\vspace{-5pt}
					\multirow{12}{15pt}{0.2} \\		
					&1  &0.000 &0.003  &0.001 &0.018  &0.001 &0.017  &0.000 &0.013 &-0.001 &0.003\\
					&2  &0.000 &0.008  &0.000 &0.047 &-0.002 &0.040  &0.001 &0.034 &-0.001 &0.008\\
					&3  &0.000 &0.012 &-0.003 &0.076  &0.003 &0.059  &0.004 &0.061  &0.000 &0.012\\
					&4  &0.001 &0.013  &0.001 &0.077  &0.001 &0.058 &-0.002 &0.059  &0.000 &0.013\\
					&5  &0.000 &0.012  &0.002 &0.079 &-0.002 &0.060 &-0.001 &0.061  &0.000 &0.012\\
					&6  &0.000 &0.012 &-0.002 &0.078  &0.000 &0.055 &-0.001 &0.061  &0.000 &0.012\\
					&7  &0.000 &0.013  &0.001 &0.078 &-0.001 &0.059 &-0.002 &0.063  &0.000 &0.013\\
					&8  &0.001 &0.012  &0.002 &0.079 &-0.001 &0.059  &0.001 &0.059  &0.000 &0.012\\
					&9  &0.000 &0.012  &0.001 &0.078  &0.000 &0.057  &0.000 &0.061  &0.000 &0.012\\
					&10 &0.000 &0.012  &0.003 &0.077  &0.001 &0.056  &0.000 &0.059 &-0.001 &0.012\\
					&11 &0.000 &0.012  &0.002 &0.075 &-0.002 &0.057  &0.002 &0.060  &0.000 &0.012\\
					\vspace{-5pt}
					\multirow{12}{15pt}{0.5} \\		
					&1   &0.000 &0.002  &0.002 &0.017  &0.001 &0.018  &0.001 &0.013 &-0.001 &0.002\\
					&2   &0.000 &0.005  &0.001 &0.048 &-0.001 &0.041  &0.000 &0.034 &-0.001 &0.005\\
					&3  &-0.001 &0.008 &-0.002 &0.075  &0.002 &0.057  &0.002 &0.061 &-0.001 &0.008\\
					&4   &0.000 &0.008 &-0.001 &0.080  &0.000 &0.057  &0.002 &0.060  &0.000 &0.008\\
					&5   &0.000 &0.008 &-0.002 &0.079 &-0.001 &0.056  &0.003 &0.060  &0.000 &0.008\\
					&6   &0.000 &0.008  &0.001 &0.077  &0.002 &0.056 &-0.001 &0.062  &0.000 &0.008\\
					&7   &0.000 &0.008  &0.000 &0.077  &0.003 &0.056  &0.002 &0.060  &0.000 &0.008\\
					&8   &0.000 &0.008  &0.002 &0.077 &-0.001 &0.056 &-0.002 &0.061  &0.000 &0.008\\
					&9   &0.000 &0.008  &0.004 &0.079  &0.001 &0.054 &-0.001 &0.059  &0.000 &0.008\\
					&10  &0.000 &0.007 &-0.003 &0.077  &0.004 &0.055 &-0.002 &0.060  &0.000 &0.008\\
					&11  &0.000 &0.008 &-0.003 &0.079  &0.001 &0.056  &0.000 &0.062  &0.000 &0.008\\
					\vspace{-5pt}
					\multirow{12}{15pt}{1} \\		
					&1   &0.000 &0.001  &0.001 &0.017  &0.001 &0.018 &0.000 &0.013 &-0.001 &0.001\\
					&2   &0.000 &0.003 &-0.001 &0.048 &-0.001 &0.039 &0.001 &0.035  &0.000 &0.003\\
					&3  &-0.001 &0.005 &-0.003 &0.079 &-0.003 &0.054 &0.001 &0.061  &0.000 &0.006\\
					&4  &-0.001 &0.006  &0.000 &0.078  &0.002 &0.057 &0.000 &0.063 &-0.001 &0.006\\
					&5   &0.000 &0.005 &-0.001 &0.079  &0.004 &0.059 &0.000 &0.060  &0.000 &0.006\\
					&6   &0.000 &0.006 &-0.004 &0.077  &0.004 &0.057 &0.002 &0.061  &0.000 &0.006\\
					&7   &0.000 &0.006  &0.004 &0.077  &0.004 &0.057 &0.001 &0.059  &0.000 &0.006\\
					&8   &0.001 &0.005 &-0.003 &0.076  &0.000 &0.058 &0.001 &0.060  &0.000 &0.006\\
					&9  &-0.001 &0.005 &-0.002 &0.080  &0.001 &0.060 &0.000 &0.061  &0.000 &0.006\\
					&10 &-0.001 &0.005 &-0.001 &0.077  &0.001 &0.057 &0.001 &0.059 &-0.001 &0.006\\
					&11  &0.000 &0.005  &0.002 &0.079  &0.003 &0.057 &0.000 &0.060  &0.000 &0.006\\
					\bottomrule
	\end{tabular}}}}
\label{tab: simul weibull vary full}
\end{center}
\end{table}
	
	Table \ref{tab: simul logistic vary full} shows that the bias of all above estimators is similar and small for all full data sizes under the logistic model. The SD of the proposed estimator is similar to that of the full data-based estimator and decreases as the full data size $N$ increases. 
	In addition, The SD of the UNI, IPW, MSCL, and ELW subsampling estimators is much larger than that of the proposed estimator, and their SD hardly changes as $N$ increases. 
	Under the Weibull model, similar observations as those under the logistic model can be observed from Table \ref{tab: simul weibull vary full}.
	
	Moving on to the inference performance of the SOS estimator, we calculated the coverage probability of the constructed confidence intervals based on the procedures behind Theorem \ref{theo: general} at a  $95\%$ confidence level. The results for different subsample sizes $n$ under the logistic and Weibull models are presented in Table \ref{tab: simul vary sub cp}.

	\begin{table}[htbp]
		\begin{center}
			\caption{The proposed estimator based coverage probability for the $j$-th dimension $\theta_{j}$ of $\btheta_{0}$ under the logistic and Weibull models with $N = 10^{6}$ and $n=k\times 10^{4}$.}{	
				\resizebox{1\columnwidth}{!}{
					{\begin{tabular}{cccccccccccccccc}
							\toprule
							&$k$& $\theta_{1}$ &$\theta_{2}$  &$\theta_{3}$ &$\theta_{4}$ &$\theta_{5}$ &$\theta_{6}$ &$\theta_{7}$ &$\theta_{8}$ &$\theta_{9}$ &$\theta_{10}$\\
							\midrule
							\multirow{4}{*}{Logistic}
							&0.5 &0.945 &0.951 &0.942 &0.952 &0.942 &0.948 &0.952 &0.939 &0.956 &0.939\\
							&1 &0.949 &0.955 &0.945 &0.952 &0.951 &0.947 &0.954 &0.939 &0.948 &0.945\\
							&2 &0.949 &0.946 &0.949 &0.944 &0.956 &0.943 &0.961 &0.943 &0.953 &0.952\\
							&5 &0.951 &0.958 &0.952 &0.940 &0.951 &0.946 &0.958 &0.944 &0.953 &0.952\\
							\midrule
							&$k$& $\theta_{1}$ &$\theta_{2}$  &$\theta_{3}$ &$\theta_{4}$ &$\theta_{5}$ &$\theta_{6}$ &$\theta_{7}$ &$\theta_{8}$ &$\theta_{9}$ &$\theta_{10}$&$\theta_{11}$\\
							\midrule
							\multirow{4}{*}{Weibull}
							&0.5&0.948 &0.955 &0.951 &0.940 &0.946 &0.948 &0.949 &0.952 &0.936 &0.958 &0.947\\
							&1&0.950 &0.954 &0.952 &0.950 &0.946 &0.955 &0.959 &0.949 &0.946 &0.956 &0.948\\
							&2&0.954 &0.952 &0.955 &0.947 &0.944 &0.957 &0.941 &0.955 &0.944 &0.949 &0.953\\
							&5&0.955 &0.953 &0.956 &0.942 &0.945 &0.958 &0.949 &0.957 &0.944 &0.947 &0.949\\
							\bottomrule
			\end{tabular}}}}
		\label{tab: simul vary sub cp}
	\end{center}
\end{table}
	
	From Table \ref{tab: simul vary sub cp}, we observe that the proposed SOS estimator achieves satisfactory coverage probabilities close to the nominal level of $95\%$ for both the logistic and Weibull models, regardless of the subsample size $n$. This indicates that the proposed method provides reliable inference for the parameter of interest.
	
	Next, we evaluate the computational performance of the proposed method.
	All the above estimators are calculated in the R programming \citep{Rcore2016} on a Windows server with a 26-core processor and 128GB RAM. The computing times of these estimators under the logistic and Weibull models for different $n$ and $N$ are presented in Table \ref{tab: time vary sub} and \ref{tab: time vary full} respectively.

	\begin{table}[htbp]
		\begin{center}
			\caption{CPU times (second) of different estimators under the logistic and Weibull models with $N = 10^{6}$ and $n=k\times10^{4}$.}
			\resizebox{.6\columnwidth}{!}{
				{\begin{tabular}{cccccccccccccccc}
						\toprule
						&$k$&UNI&IPW&MSCL&ELW&SOS\\
						\midrule
						\multirow{4}{*}{Logistic}
						&0.5 &0.251   &  0.714   &   1.067   &   0.938    &    0.451\\
						&1   &0.404   &  0.857    &  1.544       &     1.038      &0.612\\
						&2   &0.648    &  1.136    &  2.5342      &     1.305  & 0.856\\
						&5   &1.455   &   2.042    &   5.445     &       2.112   &   1.657\\
						\vspace{-8pt}\\
						\multicolumn{10}{l}{FULL:  27.920 seconds}\\
						\midrule
						&$k$&UNI&IPW&ELW&SOS\\
						\midrule
						\multirow{4}{*}{Weibull}
						&0.5 & 0.857  &   1.827    &  1.869    &   1.325\\
						&1   &1.700  &    2.908     & 2.608    &  2.168\\
						&2   &3.341  &    4.984  &     4.263   &     3.810\\
						&5  &8.529  &   11.015  &     9.216    &  9.035\\
						\vspace{-8pt}\\
						\multicolumn{10}{l}{FULL: 180.625 seconds}\\
						\bottomrule
			\end{tabular}}}
		\label{tab: time vary sub}
	\end{center}
\end{table}

\begin{table}[htbp]
\begin{center}
	\caption{CPU times (second) of different estimators under the logistic and Weibull models with $N = k\times10^{7}$ and $n=5\times10^{4}$.}
	\resizebox{.7\columnwidth}{!}{
		{\begin{tabular}{ccccccccccccccc}
				\toprule
				&$k$&FULL&UNI&IPW&MSCL&ELW&SOS\\
				\midrule
				\multirow{4}{*}{Logistic}
				&0.1 &27.920 &1.455   &   2.042    &   5.445   &   2.112  &  1.657\\
				&0.2 &57.440 &1.534 &  2.636  &  6.135   &  2.793  &  1.927\\
				&0.5 &135.240 &1.785 &   4.013  &  7.461    &  4.675 &   2.849\\
				&1   &281.020 &2.076 & 6.461 & 10.004  &   7.977 &  4.167\\
				\midrule
				&$k$&FULL&UNI&IPW&ELW&SOS\\
				\midrule
				\multirow{4}{*}{Weibull}
				&0.1 &180.625 &8.529  &   11.015  &     9.216    &  9.035\\
				&0.2 &374.510 &8.506 & 12.270   &10.424  &   9.492\\
				&0.5 &809.915&8.834  &  14.771 & 13.249&   11.270\\
				&1   &1682.105& 9.214  &  18.869 &  18.247 &  14.007\\
				\bottomrule
	\end{tabular}}}
\label{tab: time vary full}
\end{center}
\end{table}
	
	Tables \ref{tab: time vary sub} and \ref{tab: time vary full} show that the computing times of the subsampling estimators are significantly shorter compared to the full data-based estimator, particularly evident under the Weibull model and for large full data sizes. This highlights the effectiveness of subsampling methods in reducing the computational burden associated with large-scale datasets.
	Besides, the computing time of the SOS estimator is comparable to that of the UNI estimator, with the relative difference diminishing as $n$ increases. This showcases the computational efficiency of the proposed method.

	The SOS method significantly reduces the computational burden while maintaining estimation efficiency compared to the full data-based estimator,  emerging as a recommended approach to tackle challenges posed by large-scale data in practical applications.

	\section{Real data example}\label{sec: real}	
	
	In this subsection, we apply the proposed method to a large-scale airline dataset, which is available at \emph{http://stat-computing.org/dataexpo/2009/the-data.html}. The dataset encompasses flight arrival and departure details for all commercial flights within the United States from October 1987 through April 2008. The full data size is $N=116,212,331$.	
	To analyze the impact of main factors on airline delays, we employ a logistic regression model based on the dataset. We use the arrival delay indicator (1: if the arrival delay is fifteen minutes or more; 0 otherwise) as the response variable. The covariates include day/night, distance from departure to destination, weekend/weekday, and the departure delay indicator (1: if the departure delay is fifteen minutes or more; 0 otherwise).
	The full data-based estimate of the regression parameter is $(-2.824, -0.033,  1.337,  0.241,  4.127)$. The CPU time of calculating the full data-based estimator is $3762.37$ seconds.
	We take the estimation results based on the full data as a reference and then apply the proposed SOS method.
	We also calculate the UNI, IPW, and ELW estimators for comparison.
	We vary the subsample size $n$ to be $5\times 10^{3}$, $10^{4}$, $2\times 10^{4}$, $5\times 10^{4}$, $10^{5}$, $2\times 10^{5}$, $5\times 10^{5}$, and $10^{6}$.
	The results calculated based on 200 replicates are presented in Tables \ref{tab: real airline} and \ref{tab: time real}.
	
	\begin{table}[htbp]
		\begin{center}
			\caption{Bias and SD of different estimators for the $j$-th dimension of $\btheta_{0}$ under the logistic model with $N=116,212,331$ and $n=k\times 10^{4}$ in the airline data. The numbers in the table are the simulation results multiplied by $10$}{	
				\resizebox{.8\columnwidth}{!}{
					{\begin{tabular}{ccccccccccccccc}
							\toprule
							\multirow{2}{15pt}{$k$} &\multirow{2}{15pt}{{$j$}} 
							&\multicolumn{2}{c}{UNI}& \multicolumn{2}{c}{IPW}&\multicolumn{2}{c}{ELW}&\multicolumn{2}{c}{SOS}\\
							\cmidrule(lr){3-4}  \cmidrule(lr){5-6} \cmidrule(lr){7-8}\cmidrule(lr){9-10}\cmidrule(lr){11-12}
							& & Bias &SD& Bias &SD  &Bias &SD   &Bias &SD\\
							\midrule
							\multirow{5}{15pt}{0.5}	
							&1 &-0.0602 &1.4680  &0.0236 &0.7528  &0.0260 &0.8266  &0.0646 &0.0998\\
							&2 &-0.0496 &1.0277  &0.0471 &0.6160  &0.0308 &0.7389 &-0.0148 &0.0645\\
							&3  &0.0548 &5.1390 &-0.1938 &2.4830 &-0.3565 &2.5894  &0.0681 &0.3987\\
							&4 &-0.0211 &1.1199 &-0.0156 &0.6895 &-0.0004 &0.7244 &-0.0185 &0.0814\\
							&5  &0.1109 &1.2766 &-0.0566 &0.6795  &0.0420 &0.4778 &-0.1266 &0.1146\\
							\vspace{-5pt}
							\multirow{5}{15pt}{1} \\		
							&1  &0.0298 &1.0206  &0.0152 &0.5370  &0.0089 &0.5746  &0.0306 &0.0463\\
							&2 &-0.0446 &0.6887  &0.0191 &0.4444  &0.0239 &0.4888 &-0.0080 &0.0288\\
							&3  &0.0941 &3.3017 &-0.1449 &1.7176 &-0.3184 &1.7805  &0.0271 &0.1576\\
							&4 &-0.0290 &0.7697  &0.0363 &0.4831  &0.0114 &0.5007 &-0.0101 &0.0353\\
							&5  &0.0250 &0.9125 &-0.0795 &0.4590  &0.0531 &0.3358 &-0.0561 &0.0484\\
							\vspace{-5pt}
							\multirow{5}{15pt}{2} \\		
							&1 &-0.0380 &0.6893  &0.0126 &0.3915  &0.0300 &0.4034  &0.0146 &0.0231\\
							&2 &-0.0107 &0.4633  &0.0072 &0.2917 &-0.0066 &0.3726 &-0.0035 &0.0152\\
							&3  &0.2978 &2.2621 &-0.0917 &1.2679 &-0.2105 &1.1932  &0.0152 &0.0777\\
							&4 &-0.0173 &0.5446  &0.0238 &0.3484 &-0.0083 &0.3453 &-0.0052 &0.0173\\
							&5  &0.0283 &0.6527 &-0.0521 &0.3672  &0.0451 &0.2296 &-0.0278 &0.0243\\
							\vspace{-5pt}
							\multirow{5}{15pt}{5} \\		
							&1 &-0.0202 &0.4385  &0.0118 &0.2403  &0.0246 &0.2332  &0.0055 &0.0073\\
							&2 &-0.0206 &0.3279 &-0.0014 &0.1995 &-0.0190 &0.2218 &-0.0012 &0.0056\\
							&3  &0.2302 &1.4427 &-0.0049 &0.7691 &-0.0148 &0.7202  &0.0075 &0.0312\\
							&4 &-0.0017 &0.3373 &-0.0064 &0.1912 &-0.0232 &0.2360 &-0.0019 &0.0067\\
							&5 &-0.0123 &0.4007 &-0.0246 &0.2157  &0.0146 &0.1453 &-0.0109 &0.0083\\
							\vspace{-5pt}
							\multirow{5}{15pt}{10} \\	
							&1&-0.0047&0.3170&-0.0105&0.1757 &0.0117&0.1650 &0.0025&0.0045\\
							&2&-0.0044&0.2267 &0.0097&0.1388&-0.0030&0.1432&-0.0004&0.0030\\
							&3 &0.1051 &1.0185 &0.0167&0.5143&-0.0328&0.5225 &0.0039&0.0154\\
							&4&-0.0021&0.2478 &0.0031&0.1405&-0.0146&0.1807&-0.0009&0.0034\\
							&5&-0.0103&0.2788&-0.0127&0.1501 &0.0145&0.1084&-0.0053&0.0045\\
							\vspace{-5pt}
							\multirow{5}{15pt}{20} \\		
							&1 &0.0091&0.2293 &0.0086&0.1377&-0.0018&0.1284 &0.0015&0.0021\\
							&2&-0.0018&0.1675 &0.0019&0.1036 &0.0037&0.1081&-0.0003&0.0014\\
							&3 &0.0723&0.7294&-0.0017&0.3478&-0.0427&0.3679 &0.0021&0.0081\\
							&4 &0.0020&0.1791&-0.0110&0.1034 &0.0045&0.1278&-0.0006&0.0017\\
							&5&-0.0285&0.1830&-0.0176&0.1039 &0.0092&0.0738&-0.0027&0.0022\\
							\vspace{-5pt}
							\multirow{5}{15pt}{50} \\		
							&1  &0.0014 &0.1490  &0.0065 &0.0826  &0.0018 &0.0813  &0.0006 &0.0010\\
							&2  &0.0014 &0.1115  &0.0002 &0.0616  &0.0016 &0.0649 &-0.0001 &0.0007\\
							&3  &0.0357 &0.4570  &0.0128 &0.2054 &-0.0102 &0.2108  &0.0011 &0.0036\\
							&4 &-0.0017 &0.1217 &-0.0065 &0.0689 &-0.0037 &0.0783 &-0.0002 &0.0007\\
							&5 &-0.0183 &0.1191 &-0.0055 &0.0689  &0.0067 &0.0492 &-0.0012 &0.0009 \\
							\vspace{-5pt}
							\multirow{5}{15pt}{100} \\		
							&1  &0.0023 &0.0947 &-0.0011 &0.0577  &0.0023 &0.0567  &0.0004 &0.0005\\
							&2  &0.0051 &0.0784  &0.0023 &0.0430  &0.0013 &0.0488 &-0.0001 &0.0003\\
							&3  &0.0134 &0.2990  &0.0110 &0.1592 &-0.0123 &0.1624  &0.0002 &0.0015\\
							&4 &-0.0072 &0.0829  &0.0001 &0.0478 &-0.0024 &0.0524 &-0.0004 &0.0003\\
							&5 &-0.0018 &0.0807 &-0.0065 &0.0493  &0.0021 &0.0312 &-0.0006 &0.0004\\
							\bottomrule
			\end{tabular}}}}
		\label{tab: real airline}
	\end{center}
\end{table}
	
	\begin{table}[htbp]
		\begin{center}
			\caption{CPU times (second)  of different estimators under the logistic model with $N=116,212,331$ and $n=k\times10^{4}$ in the airline data.}
			\resizebox{.4\columnwidth}{!}{
				{\begin{tabular}{ccccccccccccccc}
						\toprule
						$k$&UNI&IPW&ELW&SOS\\
						\midrule
						0.5& 8.139 &   42.906   &       52.366  & 24.565\\
						1&   8.398 &   43.544    &      52.758  &  24.770\\
						2&   8.535  &  43.246    &      52.729  &  25.030\\
						5&   9.261  &   44.570   &       53.432 &  25.489\\
						10&  10.556  &   47.950    &       54.470   &  26.100\\
						20&  12.981  &  49.765     &     57.054  & 28.404\\
						50&  20.855  &  58.559     &     65.837  & 36.635\\
						100&  34.617  &   73.310    &      80.237 &  50.351\\
						\bottomrule
			\end{tabular}}}
		\label{tab: time real}
	\end{center}
\end{table}

	Table \ref{tab: real airline} provides insights into the performance of different estimators on the airline dataset. The bias of all subsampling estimators is small, particularly for large subsample sizes. The standard deviation of the IPW and ELW estimators is smaller than that of the UNI estimator, indicating the effectiveness of designing NSP and incorporating sample moments in improving the estimation efficiency of the UNI method.
	
	However, there is still considerable efficiency loss in these estimators. In contrast, the proposed SOS method outperforms other subsampling estimators in terms of standard deviation. The SD of the proposed estimator is significantly smaller than that of the IPW and ELW estimators for each component of the parameter of interest, across all subsample sizes. Moreover, the standard deviation decreases at a faster rate as the subsample size increases, indicating the superior efficiency of the proposed method.
	
	Table \ref{tab: time real} highlights the computational performance of the estimators. Subsampling methods substantially reduce the computing time compared to the full data-based estimator, with the time of $3762.37$ seconds reduced to less than one minute. The UNI estimator is the fastest among these subsampling estimators. The computing time of the SOS estimator is approximately half the time of other subsampling methods.	
	Considering both estimation efficiency and computation efficiency, the results from the numerical studies and real data analyses consistently demonstrate the significant advantages of the proposed SOS method in handling large-scale data.
	
	Next, we discuss the selection of $n$. If computational resources are highly constrained, the best we can do is to use a small $n$ that fits the available resources, and in this case, our method remains effective according to Theorem \ref{theo: general}. Otherwise, the subsample size $n$ should be chosen to strike a balance between computational and statistical efficiency. 
	The computational cost of the SOS estimator primarily arises from two aspects: computing the initial estimator and calculating the full data sample mean (i.e., the gradient in the one-step update).
	When $n$ is small, the computing time is largely dominated by the gradient computation, which is independent of $n$. In this case, the computing time of the SOS estimator has no obvious change as $n$ increases.
	When $n$ is large enough such that the computing time of the initial estimator becomes comparable to that of the full data gradient, the computing time of the SOS estimator increases notably with $n$, especially for complex models, as shown in our simulations under Weibull model (see Table \ref{tab: time vary sub}).
	On the statistical side, both our numerical results (Tables \ref{tab: simul logistic vary sub} and \ref{tab: simul Weibull vary sub}) and Corollary \ref{corr: special} suggest that the estimation efficiency of the SOS estimator improves with increasing $n$ until $n/\sqrt{N}$ becomes sufficiently large for the SOS estimator to achieve the same estimation efficiency as the full data-based estimator. After that, the efficiency gain brought by increasing $n$ is limited. This motivates us to set $n=c\sqrt{N}$ for some constant $c$. 
	The computational and statistical performance of the SOS method under different values of $c$ depends on the specific form of the function $L$ in \eqref{eq:full est} and the data generating process.
	Practitioners can select $c$ based on their experiences and domain knowledge.
	Our numerical results suggest that $c=5$ provides a good balance between computational and estimation efficiency across various scenarios. Therefore, we recommend $n=5\sqrt{N}$ as the default choice for implementing the SOS method.

	\section{Discussion}
	
	This paper proposes the SOS estimator, a simple yet useful subsampling estimator for fast statistical inference with large-scale data. The SOS estimator can achieve a faster convergence rate with the same time complexity as many existing NSP-based subsampling estimators. We obtain the asymptotic distribution of the proposed estimator in the general case with no condition on the ratio between $n$ and $N$, and provide a Monte Carlo-based procedure for statistical inference. Extensive numerical results demonstrate the promising finite-sample performance of the SOS estimator.
	
	The current SOS estimator requires the loss function to be smooth. Developing methods to handle non-smooth loss function is a meaningful topic for future research. There are some loss functions with corners or discontinuities, such as the absolute value loss, hinge loss, or the check function for quantile regression.
	In non-smooth scenarios, one could replace the gradient in the SOS method with a subgradient or some proximal step. However, convergence guarantees and rates might differ from the smooth case. The theoretical properties of SOS, which rely on a well-defined gradient and Hessian, may not hold in their current form with a non-smooth loss function. There are some possible ways to mitigate the above issue. One approach is to use a smoothed approximation of a non-smooth loss (e.g., use the convolution smoothed loss function in place of the check function for quantile regression \citep{he2023smoothed}) which can retain many benefits of the SOS method.
	
	In addition, the SOS estimator is designed for M-estimation problem with a fixed-dimensional parameter. It is of future interest to extend the SOS estimator to general semiparametric estimation problems and problems with high-dimensional parameters. In high-dimensional problems, the Hessian of the loss can be ill-conditioned. Then, the direct use of the one-step update can be less reliable. Using adaptive step sizes and dimension-reduction strategies can potentially help mitigate ill-conditioning. Moreover, many high-dimensional methods incorporate non-smooth penalties (e.g., LASSO’s $l_{1}$-penalty). Such non-smooth regularizers break the differentiability assumption required by the SOS method. Alternatives like subgradient methods, proximal algorithms, or smooth approximations can be used, but the vanilla SOS approach may require adaptation. We leave further investigations on these problems to future works.
			
	\section*{Acknowledgement}
	 This work was supported by the National Natural Science Foundation of China (Grant Number 12501405) and the Fundamental Research Fund of Beijing University of Posts and Telecommunications (Program Number 2023RC47).
		
		\newpage
		\setcounter{section}{0}
		\setcounter{condition}{0}
		\setcounter{theorem}{0}
		\setcounter{table}{0}
		\setcounter{figure}{0}
		\setcounter{example}{0}
		\setcounter{proposition}{0}
		\setcounter{lemma}{0}
		\renewcommand{\thesection}{S\arabic{section}}
		\renewcommand{\thecondition}{S\arabic{condition}}
		\renewcommand{\thetheorem}{S\arabic{theorem}}
		\renewcommand{\thetable}{S\arabic{table}}
		\renewcommand{\thefigure}{S\arabic{figure}}
		\renewcommand{\theexample}{S\arabic{example}}
		\renewcommand{\theproposition}{S\arabic{proposition}}
		\renewcommand{\thelemma}{S\arabic{lemma}}
		{\noindent\LARGE \bf Appendix: proofs for all theoretical results}

\begin{lemma}\label{lem:unif AN}
	Under Conditions (C1)--(C3), we have
	\begin{equation}\label{eq:unif AN}
		\tilde{\btheta}_{\rm uni}-\btheta_{0}= -\left[E\left\{\nabla^{2}L(\bZ;\btheta_{0})\right\}\right]^{-1}\frac{1}{n}\sum_{i\in S}\nabla L(\bZ_{i},\btheta_{0})+o_{P}\left(n^{-1/2}\right).
	\end{equation}
\end{lemma}

\begin{proof}[Proof of Lemma \ref{lem:unif AN}]
	By the definition of $\tilde{\btheta}_{\rm uni}$, we have
	\begin{equation}\label{eq: 1}
		\begin{split}
			0=\frac{1}{n}\sum_{i\in S}\nabla L(\bZ_{i};\tilde{\btheta}_{\rm uni})
			&=\frac{1}{n}\sum_{i\in S}\nabla L(\bZ_{i};\btheta_{0})+E\{\nabla^{2} L(\bZ;\btheta_{0})\}(\tilde{\btheta}_{\rm uni}-\btheta_{0})\\
			&\quad+\left[\frac{1}{n}\sum_{i\in S}\nabla^{2} L(\bZ_{i};\btheta_{0})-E\left\{\nabla^{2} L(\bZ;\btheta_{0})\right\}\right](\tilde{\btheta}_{\rm uni}-\btheta_{0})\\
			&\quad+\frac{1}{n}\sum_{i\in S}\left\{\nabla^{2} L(\bZ_{i};\bar{\btheta})-\nabla^{2}L(\bZ_{i};\btheta_{0})\right\}(\tilde{\btheta}_{\rm uni}-\btheta_{0}),
		\end{split}
	\end{equation}
	where $\bar{\btheta}$ is between $\btheta_{0}$ and $\tilde{\btheta}_{\rm uni}$. By Condition (C1), we have
	\begin{equation*}
		\left\|\frac{1}{n}\sum_{i\in S}\left\{\nabla^{2} L(\bZ_{i};\bar{\btheta})-\nabla^{2}L(\bZ_{i};\btheta_{0})\right\}\right\|\leq \frac{1}{n}\sum_{i\in S}G_{2}(\bZ_{i})\|\bar{\btheta}-\btheta_{0}\|.
	\end{equation*}
	Since $E\left\{n^{-1}\sum_{i\in S}G_{2}(\bZ_{i})\right\}=E\left\{G_{2}(\bZ)\right\}$ and $\var\left\{n^{-1}\sum_{i\in S}G_{2}(\bZ_{i})\right\}\leq n^{-1}E\left[\left\{G_{2}(\bZ)\right\}^{2}\right]$, it follows from Chebyshev's inequality that $n^{-1}\sum_{i\in S}G_{2}(\bZ_{i})=O_{P}(1)$ and hence
	\begin{equation}\label{eq: 11}
		\frac{1}{n}\sum_{i\in S}\left\{\nabla^{2} L(\bZ_{i};\bar{\btheta})-\nabla^{2}L(\bZ_{i};\btheta_{0})\right\}=O_{P}(\|\tilde{\btheta}_{\rm uni}-\btheta_{0}\|).
	\end{equation}
	Additionally, applying Chebyshev's inequality under Condition (C2), we obtain
	\begin{equation}\label{eq: 12}
		\frac{1}{n}\sum_{i\in S}\nabla^{2} L(\bZ_{i};\btheta_{0})-E[\nabla^{2} L(\bZ;\btheta_{0})]=O_{P}(n^{-1/2}).
	\end{equation}
	Then \eqref{eq: 1} together with \eqref{eq: 11} and \eqref{eq: 12} implies 
	\begin{equation}\label{eq: 2}
		-\frac{1}{n}\sum_{i\in S}\nabla L(\bZ_{i};\btheta_{0})=E\left\{\nabla^{2} L(\bZ;\btheta_{0})\right\}(\tilde{\btheta}_{\rm uni}-\btheta_{0})+O_{P}(n^{-1/2}\|\tilde{\btheta}_{\rm uni}-\btheta_{0}\|+\|\tilde{\btheta}_{\rm uni}-\btheta_{0}\|^{2}).
	\end{equation}
	By the definition of $\btheta_{0}$, we have $E\left\{\nabla L(\bZ;\btheta_{0})\right\}=0$. Then by Chebyshev's inequality, we have $n^{-1}\sum_{i\in S}\nabla L(\bZ_{i};\btheta_{0})=O_{P}(n^{-1/2})$ under Condition (C2). This together with \eqref{eq: 2} proves  $\|\tilde{\btheta}_{\rm uni}-\btheta_{0}\|=O_{P}(n^{-1/2})$ under Condition (C3), and hence the result in Lemma \ref{lem:unif AN} is proved.

\end{proof}
\begin{proof}[Proof of Theorem 1]
	In the proof, for simplicity of notation, we denote $\bv^{\otimes2}$ as the Kronecker product $\bv\otimes\bv$ for any vector $\bv$. By Taylor's expansion and some algebras, we have
	\begin{equation}\label{eq: 3}
		\begin{split}
			&\frac{1}{n}\sum_{i\in S}\nabla^{2}L(\bZ_{i};\tilde{\btheta}_{\rm uni})\left(\tilde{\btheta}_{\rm SOS}-\btheta_{0}\right)\\
			&=\frac{1}{n}\sum_{i\in S}\nabla^{2}L(\bZ_{i};\tilde{\btheta}_{\rm uni})\left(\tilde{\btheta}_{\rm uni}-\btheta_{0}\right) - \frac{1}{N}\sum_{i=1}^{N}\nabla L(\bZ_{i},\tilde{\btheta}_{\rm uni})\\
			&=\frac{1}{n}\sum_{i\in S}\nabla^{2}L(\bZ_{i};\btheta_{0})\left(\tilde{\btheta}_{\rm uni}-\btheta_{0}\right) + \left\{\frac{1}{n}\sum_{i\in S}\nabla^{2}L(\bZ_{i};\tilde{\btheta}_{\rm uni}) - \frac{1}{n}\sum_{i\in S}\nabla^{2}L(\bZ_{i};\btheta_{0})\right\}\left(\tilde{\btheta}_{\rm uni}-\btheta_{0}\right)\\
			&\quad - \frac{1}{N}\sum_{i=1}^{N}\nabla L(\bZ_{i},\btheta_{0}) - \left\{\frac{1}{N}\sum_{i=1}^{N}\nabla L(\bZ_{i},\tilde{\btheta}_{\rm uni}) - \frac{1}{N}\sum_{i=1}^{N}\nabla L(\bZ_{i},\btheta_{0})\right\}\\
			&=\frac{1}{n}\sum_{i\in S}\nabla^{2}L(\bZ_{i};\btheta_{0})\left(\tilde{\btheta}_{\rm uni}-\btheta_{0}\right) + \frac{1}{n}\sum_{i\in S}\nabla^{3}L(\bZ_{i};\breve{\btheta})\left(\tilde{\btheta}_{\rm uni}-\btheta_{0}\right)^{\otimes 2}\\
			&\quad - \frac{1}{N}\sum_{i=1}^{N}\nabla L(\bZ_{i},\btheta_{0}) - \left\{\frac{1}{N}\sum_{i=1}^{N}\nabla^{2} L(\bZ_{i},\btheta_{0})\left(\tilde{\btheta}_{\rm uni}-\btheta_{0}\right) +  \frac{1}{2N}\sum_{i=1}^{N}\nabla^{3} L(\bZ_{i},\bar{\btheta})(\tilde{\btheta}_{\rm uni}-\btheta_{0})^{\otimes2}\right\}\\
			&=\left\{\frac{1}{n}\sum_{i\in S}\nabla^{2} L(\bZ_{i},\btheta_{0})-\frac{1}{N}\sum_{i=1}^{N}\nabla^{2} L(\bZ_{i},\btheta_{0})\right\}(\tilde{\btheta}_{\rm uni}-\btheta_{0}) -\frac{1}{N}\sum_{i=1}^{N}\nabla L(\bZ_{i},\btheta_{0})\\
			&\quad + \left\{\frac{1}{n}\sum_{i\in S}\nabla^{3} L(\bZ_{i},\breve{\btheta}) - \frac{1}{2N}\sum_{i=1}^{N}\nabla^{3} L(\bZ_{i},\bar{\btheta})\right\}(\tilde{\btheta}_{\rm uni}-\btheta_{0})^{\otimes2}\\
			& = J_{1} + J_{2} + J_{3},
		\end{split}
	\end{equation}
	where $\bar{\btheta}$ and $\breve{\btheta}$ are between $\btheta_{0}$ and $\tilde{\btheta}_{\rm uni}$. 
	For $J_{1}$, by Chebyshev's inequality and Condition (C2), we have
	\begin{equation*}
		\frac{1}{n}\sum_{i\in S}\nabla^{2} L(\bZ_{i},\btheta_{0}) = E\left\{\nabla^{2} L(\bZ,\btheta_{0})\right\} + O_{P}(n^{-1/2})
	\end{equation*}
	and
	\begin{equation*}
		\frac{1}{N}\sum_{i=1}^{N}\nabla^{2} L(\bZ_{i},\btheta_{0}) = E\left\{\nabla^{2} L(\bZ,\btheta_{0})\right\} + O_{P}(N^{-1/2}).
	\end{equation*}
	Then we have
	\begin{equation*}
		\frac{1}{n}\sum_{i\in S}\nabla^{2} L(\bZ_{i},\btheta_{0})-\frac{1}{N}\sum_{i=1}^{N}\nabla^{2} L(\bZ_{i},\btheta_{0}) = O_{P}(n^{-1/2}).
	\end{equation*}
	Recalling the result \eqref{eq:unif AN} in Lemma \ref{lem:unif AN}, we have
	\begin{equation}\label{eq: 1 r}
		\begin{split}
			J_{1} 
			&=  \left\{\frac{1}{n}\sum_{i\in S}\nabla^{2} L(\bZ_{i},\btheta_{0})-\frac{1}{N}\sum_{i=1}^{N}\nabla^{2} L(\bZ_{i},\btheta_{0})\right\}\left\{-\left[E\left\{\nabla^{2}L(\bZ;\btheta_{0})\right\}\right]^{-1}\frac{1}{n}\sum_{i\in S}\nabla L(\bZ_{i},\btheta_{0})\right\}\\
			&\quad+o_{P}\left(n^{-1}\right)
		\end{split}
	\end{equation}
	For $J_{3}$, we have
	\begin{equation}\label{eq: J3 process}
		\begin{split}
			&\Big\|\frac{1}{N}\sum_{i=1}^{N}\nabla^{3} L(\bZ_{i},\bar{\theta})-E[\nabla^{3} L(\bZ,\theta_{0})]\Big\|\\
			&\leq \Big\|\frac{1}{N}\sum_{i=1}^{N}\left\{\nabla^{3} L(\bZ_{i},\bar{\theta})-\nabla^{3} L(\bZ_{i},\theta_{0})\right\}\Big\| +\Big\|\frac{1}{N}\sum_{i=1}^{N}\nabla^{3} L(\bZ_{i},\theta_{0})-E[\nabla^{3} L(\bZ,\theta_{0})]\Big\|\\
			& = \frac{1}{N}\sum_{i=1}^{N}\big\|\nabla^{3} L(\bZ_{i},\bar{\theta})-\nabla^{3} L(\bZ_{i},\theta_{0})\big\|+O_{P}(N^{-1/2})\\
			&\leq \frac{1}{N}\sum_{i=1}^{N}G_{3}(Z)\|\bar{\theta}-\theta_{0}\|+O_{P}(N^{-1/2})\\
			&\leq E\{G_{3}(Z)\}\|\tilde{\theta}_{\rm uni}-\theta_{0}\|+O_{P}(N^{-1/2})\\
			& = O_{P}(n^{-1/2}).
		\end{split}
	\end{equation}
	The first equality in \eqref{eq: J3 process} follows from Condition (C2) and Chebyshev's inequality, the second inequality in \eqref{eq: J3 process} is derived using Condition (C1), the third inequality in \eqref{eq: J3 process} is obtained using Condition (C1) and the relationship between $\bar{\btheta}$ and $\tilde{\btheta}_{\rm uni}$, and the last equality is established by applying Lemma \ref{lem:unif AN}.
	Then we have
	\begin{equation}\label{eq: 31}
		\frac{1}{N}\sum_{i=1}^{N}\nabla^{3} L(\bZ_{i},\bar{\theta})
		=E[\nabla^{3} L(\bZ,\theta_{0})]+O_{P}(n^{-1/2}).
	\end{equation}
	Through a similar proof process as that used in proving \eqref{eq: 31}, we can show that 
	\begin{equation}\label{eq: 32}
		\frac{1}{n}\sum_{i\in S}\nabla^{3} L(\bZ_{i},\breve{\btheta})=E[\nabla^{3} L(\bZ_{i},\btheta_{0})]+O_{P}(n^{-1/2}).
	\end{equation} 
	Then combing \eqref{eq: 31} and \eqref{eq: 32}, we have 
	\begin{equation}\label{eq: 31 res}
		\frac{1}{n}\sum_{i\in S}\nabla^{3} L(\bZ_{i},\breve{\btheta}) - \frac{1}{2N}\sum_{i=1}^{N}\nabla^{3} L(\bZ_{i},\bar{\btheta}) = \frac{1}{2}E[\nabla^{3} L(\bZ_{i},\btheta_{0})]  +  O_{P}(n^{-1/2}).
	\end{equation}
	In addition, by Conditions (C2),(C3) and Chebyshev's inequality, we have
	\begin{equation*}
		\left[E\left\{\nabla^{2}L(\bZ;\btheta_{0})\right\}\right]^{-1}\frac{1}{n}\sum_{i\in S}\nabla L(\bZ_{i},\btheta_{0}) = O_{P}(n^{-1/2}).
	\end{equation*}
	Then recalling the result \eqref{eq:unif AN} in Lemma \ref{lem:unif AN}, we have
	\begin{equation}\label{eq: 32 res}
		(\tilde{\btheta}_{\rm uni}-\btheta_{0})^{\otimes2} 
		= \left(\left[E\left\{\nabla^{2}L(\bZ;\btheta_{0})\right\}\right]^{-1}\frac{1}{n}\sum_{i\in S}\nabla L(\bZ_{i},\btheta_{0})\right)^{\otimes2}  + o_{P}(n^{-1}).
	\end{equation}
	Then \eqref{eq: 32 res} together with \eqref{eq: 31 res} implies
	\begin{equation}\label{eq: 3 r}
		\begin{split}
			J_{3} 
			&=  \left[E\left\{\nabla^{3} L(\bZ_{i},\btheta_{0})\right\}/2\right]\left(\left[E\left\{\nabla^{2}L(\bZ;\btheta_{0})\right\}\right]^{-1}\frac{1}{n}\sum_{i\in S}\nabla L(\bZ_{i},\btheta_{0})\right)^{\otimes2}\\
			&\quad+o_{P}\left(n^{-1}\right)
		\end{split}
	\end{equation}
	Then by plugging the expansions of $J_{1}$ in \eqref{eq: 1 r} and $J_{3}$ in \eqref{eq: 3 r} to the rightmost side of the Equation \eqref{eq: 3}, we have
	\begin{equation}\label{eq: comp}
		\begin{split}
			&\frac{1}{n}\sum_{i\in S}\nabla^{2}L(\bZ_{i};\tilde{\btheta}_{\rm uni})\left(\tilde{\btheta}_{\rm SOS}-\btheta_{0}\right)\\
			&=\left\{\frac{1}{n}\sum_{i\in S}\nabla^{2} L(\bZ_{i},\btheta_{0})-\frac{1}{N}\sum_{i=1}^{N}\nabla^{2} L(\bZ_{i},\btheta_{0})\right\}\left\{-\left[E\left\{\nabla^{2}L(\bZ;\btheta_{0})\right\}\right]^{-1}\frac{1}{n}\sum_{i\in S}\nabla L(\bZ_{i},\btheta_{0})\right\}\\
			&\quad -\frac{1}{N}\sum_{i=1}^{N}\nabla L(\bZ_{i},\btheta_{0})\\
			&\quad + \left[E\left\{\nabla^{3} L(\bZ_{i},\btheta_{0})\right\}/2\right]\left(\left[E\left\{\nabla^{2}L(\bZ;\btheta_{0})\right\}\right]^{-1}\frac{1}{n}\sum_{i\in S}\nabla L(\bZ_{i},\btheta_{0})\right)^{\otimes2}+o_{P}(n^{-1}).
		\end{split}
	\end{equation}
	Through a similar proof process as that used in proving \eqref{eq: 31}, we can show that 
	\begin{equation}\label{eq: comp 0 res}
		\begin{split}
			\frac{1}{n}\sum_{i\in S}\nabla^{2}L(\bZ_{i};\tilde{\btheta}_{\rm uni})=E\left\{\nabla^{2}L(\bZ;\btheta_{0})\right\}+O_{P}(n^{-1/2}).
		\end{split}
	\end{equation}
	Then, by Slutsky's theorem, we have
	\begin{equation}
		\begin{split}
			&\tilde{\btheta}_{\rm SOS}-\btheta_{0}\\
			&=\bbH^{-1}\left\{\frac{1}{n}\sum_{i\in S}\nabla^{2} L(\bZ_{i},\btheta_{0})-\frac{1}{N}\sum_{i=1}^{N}\nabla^{2} L(\bZ_{i},\btheta_{0})\right\}\left\{-\bbH^{-1}\frac{1}{n}\sum_{i\in S}\nabla L(\bZ_{i},\btheta_{0})\right\}\\
			&\quad -\bbH^{-1}\frac{1}{N}\sum_{i=1}^{N}\nabla L(\bZ_{i},\btheta_{0})\\
			&\quad + \bbH^{-1}\left(\bbM/2\right)\left\{\bbH^{-1}\frac{1}{n}\sum_{i\in S}\nabla L(\bZ_{i},\btheta_{0})\right\}^{\otimes2}+o_{P}(n^{-1}),
		\end{split}
	\end{equation}
	where $\bbH = E\left\{\nabla^{2}L(\bZ;\btheta_{0})\right\}$ and $\bbM = E\left\{\nabla^{3} L(\bZ,\btheta_{0})\right\}$.
	
	Recalling the definitions $c_1 = \lim_{n,N\to \infty}\min(\sqrt{N},n)/n$ and $c_{2} =  \lim_{n,N\to \infty}\min(\sqrt{N},n)/\sqrt{N}$ in Theorem 1, we have
	\begin{equation}
		\begin{split}
			&\min(\sqrt{N},n)\left(\tilde{\btheta}_{\rm SOS}-\btheta_{0}\right)\\
			&=- c_{1}\bbH^{-1}\left[\sqrt{n}\left\{\frac{1}{n}\sum_{i\in S}\nabla^{2} L(\bZ_{i},\btheta_{0})-\frac{1}{N}\sum_{i=1}^{N}\nabla^{2} L(\bZ_{i},\btheta_{0})\right\}\right]\left[\bbH^{-1}\left\{\sqrt{n}\frac{1}{n}\sum_{i\in S}\nabla L(\bZ_{i},\btheta_{0})\right\}\right]\\
			&\quad -c_{2}\bbH^{-1}\sqrt{N}\frac{1}{N}\sum_{i=1}^{N}\nabla L(\bZ_{i},\btheta_{0})\\
			&\quad+ c_{1}\bbH^{-1}\left(\bbM/2\right)\left\{\bbH^{-1}\sqrt{n}\frac{1}{n}\sum_{i\in S}\nabla L(\bZ_{i},\btheta_{0})\right\}^{\otimes 2} + o_{P}\left(1\right).
		\end{split}
	\end{equation}
	
	Let $\bU_{N} = (\bU_{N,1}^{\T},\bU_{N,2}^{\T},\bU_{N,3}^{\T})^{\T}$ be a $2d+(d^2-d)/2+d$ dimensional random vector sequence with $\bU_{N,1} = 1/\sqrt{n}\sum_{i\in S}\nabla L(\bZ_{i},\btheta_{0})$, $\bU_{N,2} = 1/\sqrt{N}\sum_{i=1}^{N}\nabla L(\bZ_{i},\btheta_{0})$, $\bU_{N,3}$ being the vectorized form of the upper triangle of 
	$$\sqrt{n}\left\{\frac{1}{n}\sum_{i\in S}\nabla^{2} L(\bZ_{i},\btheta_{0})-\frac{1}{N}\sum_{i=1}^{N}\nabla^{2} L(\bZ_{i},\btheta_{0})\right\},$$
	and
	\begin{equation}
		\bg(\bU_{N}) 
		= c_{1}\bbH^{-1}(\bbM/2)(\bbH^{-1}\bU_{N,1})^{\otimes 2}-c_{2}\bbH^{-1}\bU_{N,2} -c_{1}\bbH^{-1}\bU_{C,N}(\bbH^{-1}\bU_{N,1}),
	\end{equation}
	with $\bbU_{C,N}$ being a $d\times d$ symmetrical matrix and the upper triangle matrix of $\bbU_{C,N}$ consisting of the elements in $\bU_{N,3}$ arranged in rows. Then we have
	\begin{equation}
		\min(\sqrt{N},n)\left(\tilde{\btheta}_{\rm SOS}-\btheta_{0}\right)=\bg(\bU_{N}) + o_{P}\left(1\right).
	\end{equation}
	
	To prove the result in Theorem 1, according to the continuous mapping in Theorem 2.3 of \cite{Vaart2000AS}, it suffices to prove that
	$\bU_{N}\stackrel{d}{\to}N(0,\bbV)$. 
	Note that $\bU_{N}$ can be written as 
	\begin{equation}
		\bU_{N} = \frac{1}{N}\sum_{i=1}^{N}
		\begin{pmatrix}
			\sqrt{n}\frac{R_{i}}{n/N}\nabla L(\bZ_{i},\btheta_{0})\\
			\sqrt{N}\nabla L(\bZ_{i},\btheta_{0})\\
			\sqrt{n}\left(\frac{R_{i}}{n/N}-1\right)vec[upper\{\nabla^{2} L(\bZ_{i},\btheta_{0})\}]
		\end{pmatrix}
		\triangleq \frac{1}{N}\sum_{i=1}^{N}\bK_{i},
	\end{equation}
	where $vec(\bbA)$ and $upper(\bbA)$ respectively denote the vectorized form and the upper triangle matrix of $\bbA$. Note that $N^{-2}\sum_{i=1}^{N}\Cov(\bK_{i})\to \bbV$. By applying the Lindeber-Feller central limit theorem in Proposition 2.27 of \cite{Vaart2000AS}, to prove $\bU_{N}\stackrel{d}{\to}N(0,\bbV)$, it suffices to verify the Lindeberg condition
	\begin{equation}\label{condi1}
		\frac{1}{N^{2}}\sum_{i=1}^{N}E[\|\bK_{i}\|^{2}1\{\|\bK_{i}\|>N\epsilon\}]\to 0
	\end{equation}
	for every $\epsilon$.	
	Since $E[\|\bK_{i}\|^{2}1\{\|\bK_{i}\|>N\epsilon\}]\leq E[\|\bK_{i}\|^{2+\tau}/(N\epsilon)^{\tau}]$ holds for any $\tau>0$, then proving \eqref{condi1} reduces to showing that
	\begin{equation*}\label{condi: 1}
		E[\|\bK_{i}\|^{2+\tau}]=o(N^{1+\tau})
	\end{equation*}
	for some $\tau>0$.
	By the definition of $\bK_{i}$ and Condition (C2), we have
	\begin{equation*}
		E[\|\bK_{i}\|^{2+\tau}]=O(N^{1+\tau}/n^{\tau/2}+N^{1+\tau/2}) = o(N^{1+\tau}).
	\end{equation*}
\end{proof}

\begin{proof}[Proof of Corollary 1]
	
	Similar to Lemma \ref{lem:unif AN}, under Conditions (C1)--(C3), we have
	\begin{equation}\label{eq:full AN}
		\hat{\btheta}_{\rm full}-\btheta_{0}= -\left[E\left\{\nabla^{2}L(\bZ;\btheta_{0})\right\}\right]^{-1}\frac{1}{N}\sum_{i=1}^{N}\nabla L(\bZ_{i},\btheta_{0})+o_{P}\left(N^{-1/2}\right).
	\end{equation}
	Then by the central limit theorem, we have $\sqrt{N}(\hat{\btheta}_{\rm full}-\btheta_{0})\stackrel{d}{\to}N(0,\bbV_{c})$.
	Then to prove the results in Corollary 1, it suffices to prove that $\sqrt{N}(\tilde{\btheta}_{\rm SOS}-\hat{\btheta}_{\rm full})\stackrel{p}{\to}0$. 
	Note that
	\begin{equation}
		\begin{split}
			&\tilde{\btheta}_{\rm SOS}-\hat{\btheta}_{\rm full}\\
			&= \left[\left\{\frac{1}{n}\sum_{i\in S}\nabla^2 L(\bZ_{i};\tilde{\btheta}_{\rm uni})\right\}^{-1}-\left\{\frac{1}{N}\sum_{i=1}^{N}\nabla^2 L(\bZ_{i};\bar{\btheta})\right\}^{-1}\right]\frac{1}{N}\sum_{i=1}^{N}\nabla L(\bZ_{i};\tilde{\btheta}_{\rm uni})\\
			&=T_{1}T_{2},
		\end{split}
	\end{equation}
	where  $\bar{\btheta}$ is between $\tilde{\btheta}_{\rm uni}$ and $\hat{\btheta}_{\rm full}$.	
	For $T_{1}$,  by Conditions (C1), (C2) and Chebyshev's inequality, we have $n^{-1}\sum_{i\in S}\nabla^2 L(\bZ_{i};\tilde{\btheta}_{\rm uni})=E\{\nabla^{2}L(\bZ;\btheta_{0})\}+O_{P}(n^{-1/2})$ and $N^{-1}\sum_{i=1}^{N}\nabla^2 L(\bZ_{i};\bar{\btheta})=E\{\nabla^{2}L(\bZ;\btheta_{0})\}+O_{P}(n^{-1/2})$.
	Then by Condition (C3), we have
	$T_{1} = O_{P}(n^{-1/2})$.
	In addition, by the definition of $\tilde{\btheta}_{\rm uni}$ and Conditions (C1),(C2), we have $T_{2} = O_{P}(n^{-1/2})$.
	Then we have $\tilde{\btheta}_{\rm SOS}-\hat{\btheta}_{\rm full}=O_{P}(n^{-1})=o_{P}(N^{-1/2})$.
	
\end{proof}
		
		\bibliographystyle{asa}
		\bibliography{SOS-arXiv.bib}
\end{document}